\documentclass[11pt]{article}
\usepackage{amsbsy,amsmath,amsthm,amssymb,times,subfigure,graphics,graphicx,bm,mathtools,xcolor}
\RequirePackage{natbib}
\usepackage{footnote}
\usepackage{multirow}
\usepackage{url}
\usepackage[utf8]{inputenc}
\usepackage{algorithm}
\usepackage{algpseudocode}
\usepackage{pslatex} 
\usepackage{commath}
\usepackage{caption}
\usepackage{mathrsfs}
\allowdisplaybreaks

\newtheorem{theorem}{Theorem}
\newtheorem{lemma}{Lemma}
\newtheorem{corollary}{Corollary}

\theoremstyle{definition}

\newtheorem{remark}{Remark}

\newcommand{\I}{{\mathrm{I}}}
\newcommand{\tr}{{\mathrm{tr}}}

\newcommand{\E}{{\mathrm{E}}}
\newcommand{\Var}{{\mathrm{Var}}}
\newcommand{\rank}{{\mathrm{rank}}}

\newcommand{\pcite}[1]{\citeauthor{#1}'s \citeyearpar{#1}}

\begin{document}

\title{Necessary and sufficient conditions for posterior propriety for generalized linear mixed models}

\author{Yalin Rao$^1$ and Vivekananda Roy$^2$\\
  Department of Mathematics and Statistics, University of Massachusetts, Amherst$^1$\\ Department of Statistics, Iowa State University$^2$}

 \date{}

 \maketitle
 
  \begin{abstract}
    Generalized linear mixed models (GLMMs) are commonly used to
    analyze correlated discrete or continuous response data. In
    Bayesian GLMMs, the often-used improper priors may yield
    undesirable improper posterior distributions. Thus, verifying
    posterior propriety is crucial for valid applications of Bayesian
    GLMMs with improper priors. Here, we consider the popular improper
    uniform prior on the regression coefficients and several proper or
    improper priors, including the widely used gamma and power priors
    on the variance components of the random effects. We also
    construct an approximate Jeffreys' prior for objective Bayesian
    analysis of GLMMs. For the two most widely used GLMMs, namely, the
    binomial and Poisson GLMMs, we provide easily verifiable
    sufficient conditions compared to the currently available
    results. We also derive the necessary conditions for posterior
    propriety for the general exponential family GLMMs. Finally, we
    use examples involving one-way and two-way random effects models
    to demonstrate the theoretical results derived here.
\end{abstract}


{\it Keywords: Bayesian GLMMs, Diffuse prior, Improper prior, Jeffreys' prior, Objective Bayes, Reference prior, Variance component}


\section{Introduction} 
\label{sec:int}
Generalized linear mixed models (GLMMs) are widely used statistical
models. The popularity of GLMMs is due to their flexibility in
incorporating both fixed and random effects and the ability to handle
non-normal and heteroscedastic data. The random effects in
GLMMs are particularly useful for capturing correlation in data sets
involving repeated measures, longitudinal studies, or other complex
relationships.

Bayesian GLMMs allow for the incorporation of prior knowledge, if
available, about the parameters and making probabilistic statements on
them. On the other hand, if improper priors are used for
expressing minimal prior information, or for mathematical convenience, or for sensitivity analysis to assess how the choice of prior impacts
the resulting inference, the resulting posterior distributions of the
parameters are not guaranteed to be proper. It is known that the
likelihood function of a GLMM is not available in closed form, and so
are the posterior densities of the parameters. Generally, Markov chain
Monte Carlo (MCMC) algorithms are used to explore the posterior
densities of GLMMs \citep{roy:2022}. However, it is known that the
standard MCMC estimator converges to zero with probability one if the
MCMC chain corresponds to an improper posterior distribution
\citep{athreya2014monte}. Further, MCMC samplers may not provide a red flag when the posterior distribution is
improper \citep{hobert1996effect}. Thus, one must undertake
theoretical analysis to establish posterior propriety before
proceeding with inference. In this article, we study conditions
guaranteeing posterior propriety for GLMMs for widely used improper
priors on the fixed effects and the variance components parameters.

A few articles in the literature address the posterior propriety of
Bayesian GLMMs. To briefly describe the results of these
articles, we introduce some standard notations. In particular, we
denote the observed data by $y=(y_{1},\dots,y_{n})$. Also, we use $X$
and $Z$ to denote the $n \times p$ fixed effect design matrix and the
$n \times q$ random effects design matrix, respectively. Finally, let
$\beta \in \mathbb{R}^{p}$ be the regression coefficients vector, and
$u \in \mathbb{R}^{q}$ be the random effects vector following a
multivariate normal distribution. \cite{nata:mccu:1995} derive
necessary and sufficient conditions for the posterior propriety for a
single variance component Bernoulli GLMM under the assumption that
$\beta$ is known. In \cite{natarajan2000reference}, sufficient and
necessary conditions for posterior propriety for the Bernoulli GLMM
with a general covariance matrix for the random effect vector $u$ and
the improper uniform prior on $\beta$ have been studied.
\cite{natarajan2000reference} also show that for GLMMs, an improper prior on
the covariance matrix of $u$ obtained by Jefferys' rule results in an
improper posterior for any prior on $\beta$. 
In
\cite{chen2002necessary}, sufficient conditions for the posterior
propriety include that for at least $p$ observations, $y_i>0$ for the
Poisson family and $0<y_i<m_i$ for the binomial family where $m_i$ is
the maximum possible value for the $i$th binomial random variable.
\cite{chen2002necessary} also assume that the sub-matrix of $X$
corresponding to the observations satisfying these conditions has full
rank. In \cite{natarajan2000reference}, sufficient conditions for the
posterior propriety for Poisson GLMMs include these conditions. In
general, \pcite{natarajan2000reference} conditions require some
complex integrals to be finite verifying which seems difficult.

On the other
hand, in Sections \ref{sec:suffi} and \ref{sec:nece} of this article,
we present easily verifiable necessary and sufficient conditions for
posterior propriety for the binomial and Poisson GLMMs.
Sufficient conditions for posterior propriety for a binary GLMM with a
particular prior on the variance of the random effects have been derived in
\pcite{chen2002necessary} Section 4.2. One of their conditions is that
$Z$ is of full column rank. \cite{mich:morr:2016} study posterior propriety
for GLMMs when an `exponentiated norm bound' condition holds for the
likelihood function. \cite{mich:morr:2016} mention that if the likelihood is a log-concave function, and the maximum
likelihood estimator (MLE) exists and is unique, then the
exponentiated norm bound condition holds. 
But, as proved in
\cite[][Theorem 3.1]{chen2001propriety} for binary models, the
necessary and sufficient conditions of posterior propriety and that of
the existence of MLE overlap. Thus, the exponentiated norm bound
condition assumes as much as the conclusion. Further, unless the prior
on the covariance matrix of the random effects is proper, \pcite{mich:morr:2016} conditions require that
$Z$ is of full column rank.

Geometric ergodicity of the Gibbs samplers for
binary logistic mixed models, binary probit mixed models, and normal
linear mixed models with the improper uniform prior on the regression
coefficients and proper or improper priors on the variance components,
has been explored in \cite{rao:roy:2021}, \cite{wang2018convergence}
and \cite{roman2012convergence}, respectively. Note that the geometric
ergodicity of the Markov chain implies that the invariant posterior
density is proper. In all of these papers, the main results on
geometric ergodicity include different conditions on the random effect
matrix $Z$. Furthermore, in \cite{wang2018convergence}, the sufficient
conditions for posterior propriety for binary GLMMs require a matrix
closely related with $Z$ is of full column rank. Here, some sets of our
sufficient conditions on posterior propriety of binomial and Poisson
GLMMs do {\it not} put any restriction on the $Z$ matrix. When $Z$ is
a matrix consisting of only zeros and ones, which is often the case in
practice, the assumption of the full column rank of $Z$ may not hold. Indeed,
in Section \ref{sec:application} of this article, we provide an
example where $Z$ is not of full rank, but the results of this paper
can be used to establish the posterior propriety.

The rest of the manuscript is organized as follows. In
Section \ref{sec:suffi}, we provide sufficient conditions for
posterior propriety for the binomial and Poisson GLMMs. In Section
\ref{sec:nece}, we present the necessary conditions for posterior
propriety for the binomial data, including the special case of binary
data under both cases when $Z$ is full rank or not full rank. In addition, Section
\ref{sec:nece} includes the necessary conditions
for posterior propriety for the general exponential family GLMMs. In
Section \ref{sec:jp}, we derive an approximate Jeffreys' prior for the
parameters in GLMMs. In Section~\ref{sec:application}, we consider examples of one-way and
two-way random effects models and check posterior propriety
by employing the conditions derived in this article. Finally, some
concluding remarks appear in Section~\ref{sec:disc}. Proofs of most of
the theoretical results are given in the Appendix.

\section{Sufficient conditions for posterior propriety for some popular GLMMs} 
\label{sec:suffi}

This section considers posterior propriety for the two most
common non-Gaussian GLMMs: binomial and Poisson GLMMs. Suppose
$x_{i}^\top$ and $z_{i}^\top$ indicate the $i$th row of $X$ and $Z$,
respectively, $i = 1,...,n$. In the case when $Y_{1},Y_{2},...,Y_{n}$
are binomial random variables, we assume the conditional mean
$\E(Y_{i}| \beta, u)$ is related to the linear predictor
$\eta_i = x_{i}^\top \beta + z_{i}^\top u$ in the usual way, that is  
$\E(Y_{i}| \beta, u)= m_{i}F(\eta_{i})$, where $F$ is a cumulative distribution
function (cdf). The two most frequently used binomial GLMMs are
logistic and probit GLMMs, where $F(\cdot)$ is the cdf of the standard
logistic and the standard normal random variables, respectively. If
$(Y_{1},Y_{2},...,Y_{n})$ is the vector of Poisson response variables, we assume the popular log link, that is $\E (Y_{i}| \beta, u) = \exp(\eta_{i})$. Thus, for the binomial GLMMs,
we have
\begin{equation}
\label{eq:disu}
 Y_{i} \mid \beta,u \overset{ind}{\sim} Binomial(m_{i}, F(x_{i}^\top \beta + z_{i}^\top u))  \; \text{for} \; i = 1,\dots,n, 
\end{equation}
and for the Poisson GLMMs, we have
\begin{equation}
\label{eq:dis_poisson}
 Y_{i} \mid \beta,u \overset{ind}{\sim} Poisson(\exp[x_{i}^\top \beta +
z_{i}^\top u])  \;  \text{for} \; i = 1,\dots,n. 
\end{equation} 
Without loss of generality, we assume there are $r$ random effects
$u_{1}^\top, u_{2}^\top,\dots,u_{r}^\top$, where $u_{j}$ is a
$q_{j} \times 1$ vector with $q_{j} > 0$, and
$q_{1} + q_{2} + \dots+q_{r} = q$. Let
$u_{j} \overset{ind}{\sim} N(0, (1/\tau_{j}) \I_{q_{j}})$, where
$\tau_{j} > 0$ for $j=1,\dots,r$. Thus,
$u = (u_{1}^\top,\dots,u_{r}^\top)^\top$ and let
$\tau = (\tau_{1},\dots,\tau_{r})$.

The likelihood function for $(\beta,\tau)$ is
\begin{align}
\label{eq:general_likef}
L(\beta,\tau \mid y)= \int_{\mathbb{R}^{q}}\prod_{i=1}^{n} p(y_{i} \mid \beta,u) \phi_{q}(u;0,D(\tau)^{-1}) du,
\end{align}
where $p(y_{i} \mid \beta,u)$ is the pmf of $y_i$ and
$D(\tau)^{-1} = \oplus_{j=1}^{r} (1/\tau_{j}) \I_{q_{j}}$ with
$\oplus$ denoting the direct sum. Also, $\phi_{q}(u;0,D(\tau)^{-1})$ denotes the probability density function of
the $q$-dimensional normal distribution with mean vector $0$,
covariance matrix $D(\tau)^{-1}$ evaluated at $u$. For binomial GLMMs, $p(y_{i} \mid \beta,u)$ follows from
\eqref{eq:disu} and for Poisson GLMMs, $p(y_{i} \mid \beta,u)$ follows from \eqref{eq:dis_poisson}.

We consider an improper prior on $\beta$,  $\pi(\beta) \propto 1 $. The prior for $\tau_{j}$ is 
\begin{align}
\label{eq:tauprior}
\pi(\tau_{j}) \propto \tau_{j}^{a_{j}-1}e^{-\tau_{j}b_{j}}, \;j = 1,...,r,
\end{align} 
which may be proper or improper depending on the values of $a_{j}$ and
$b_{j}$. Also, we assume that $\beta$ and $\tau$ are apriori
independent and all the $\tau_{j}$s are also apriori
independent. Thus, the posterior density of $(\beta, \tau)$ is
$ \pi(\beta,\tau \mid y) \propto L(\beta,\tau \mid y)\pi(\tau)$ where
$\pi(\tau) = \prod_{j=1}^r \pi(\tau_j)$. Here, we will derive the
conditions on the values of $a_{j}$ and $b_{j}$ for guaranteeing
propriety of $ \pi(\beta,\tau \mid y)$.

\subsection{Sufficient conditions for posterior propriety for binomial data}   
\label{subs:binomial}
For binomial data, the likelihood function for $(\beta, \tau)$ in \eqref{eq:general_likef} becomes
\begin{align*}
L(\beta,\tau \mid y)= \int_{\mathbb{R}^{q}}\bigg[\prod_{i=1}^{n} {m_{i} \choose y_{i}} \big[F(x_{i}^\top \beta + z_{i}^\top u)\big{]}^{y_{i}}\big{[}1- F(x_{i}^\top \beta + z_{i}^\top u) \big]^{m_{i}-y_{i}} \bigg]\phi_{q}(u;0,D(\tau)^{-1}) du.
\end{align*}
Thus the posterior density $
\pi(\beta,\tau \mid y) = L(\beta,\tau \mid y)\pi(\tau)/c(y)$ is proper if and only if 
\begin{align}
\label{eq:cy}
&c(y) = \int_{\mathbb{R}_{+}^{r}} \int_{\mathbb{R}^{p}} \int_{\mathbb{R}^{q}} \bigg[\prod_{i=1}^{n} {m_{i} \choose y_{i}} \big{[}F(x_{i}^\top \beta +z_{i}^\top u)\big{]}^{y_{i}}\big[1- F(x_{i}^\top \beta + z_{i}^\top u) \big]^{m_{i}-y_{i}}\bigg]\nonumber\\  & \hspace{1.5in} \times \phi_{q}(u;0,D(\tau)^{-1}) du \pi(\tau)d\beta d\tau < \infty.
\end{align}
Before we present the results for binomial data, we introduce some
notations. Let $N = \{1,2,...,n\}$. As in \cite{roy2013posterior} \cite[see also][]{chen2004propriety}, we
partition the index set as
$N = I_{1} \biguplus I_{2} \biguplus I_{3}$, where we define
$I_{1} =\{i \in N: y_{i} =0\}, \,I_{2} = \{i \in N: y_{i} = m_{i}\}$,
$I_{3} = \{i \in N: 1 \le y_{i} \le m_{i} -1\}$ and $k$ is the
cardinality of $I_{3}$. Recall that $X$ is the $n \times p$ design matrix
with the $i$th row $x_{i}^\top$. Let $\breve{X}$ be the $k \times p$
matrix with rows $x_{i}^\top$ where $i \in I_{3}$ and the $(n+k) \times p$ matrix
$X_{\triangle} =( X^\top : \breve{X}^\top)^\top$ with the
$i$th row $x_{\triangle i}^\top$. Define $X^*_\triangle$ be the
$(n+k) \times p$ matrix with the $i$th row as
$t_{i} x_{\triangle i}^\top$ where $t_{i}=1$ if
$i \in I_{1}\cup I_{3}$, $t_{i} = -1$ if $i \in I_{2}$ and
$t_{n+j} = -1$ for $j=1,\dots,k$. 
The following Theorem states the
conditions for posterior propriety for binomial data when the prior
for $\beta$ is $\pi(\beta) \propto 1 $ and the prior for $\tau_{j}$ is
as in \eqref{eq:tauprior}.
\begin{theorem}
\label{theosuf}
Assume the following conditions are satisfied
\begin{enumerate}
\item $X$ is of full column rank, and there exists a positive vector $e>0$ such that $e^{\top} X^*_\triangle = 0$;
\item $a_{j} > p/2,\,b_{j} > 0$ for $j = 1,...,r$;
\item    $\E | \delta_\circ|^{p} < \infty$, where $\delta_\circ \sim F$.
\end{enumerate}
Then \eqref{eq:cy} holds.
\end{theorem}
\begin{remark}
  The condition $1$ in Theorem~\ref{theosuf} can be easily verified by
  a simple optimization method \citep{roy2007convergence} using
  publicly available software packages (see
  Section~\ref{sec:application} for details). Also, for the probit or
  logistic GLMMs, the condition $3$ is satisfied as moments of all
  orders for the normal and logistic random variables are finite.
\end{remark}
\begin{remark}
  Posterior propriety for nonlinear mixed models is studied in
  \cite{chen2002necessary}. For binomial data, the sufficient
  conditions in that paper include $0<y_{i}<m_{i}$ for at least $p$
  observations with the corresponding sub-matrix of $X$ being full
  rank. These conditions do not hold for the examples in
  Section~\ref{sec:application}, whereas the conditions of
  Theorem~\ref{theosuf} are
  satisfied. 
\end{remark}
\begin{remark}
  \cite{sun2001propriety} investigate the necessary and sufficient
  conditions for posterior propriety for linear mixed models. In their
  Theorem $2$, the sufficient conditions include rank conditions for
  $Z$ and some conditions on $q_{j}$ while we do not have any
  conditions on $Z$ or $q_j$ in Theorem \ref{theosuf}.
\end{remark}
\begin{remark}
  Sufficient conditions for geometric ergodicity of certain Gibbs
  samplers have been established for linear mixed models
  \citep{roman2012convergence} and GLMMs \citep{wang2018convergence,
    rao:roy:2021}. These conditions also imply posterior propriety, and
  the conditions involve $Z$ and $q_{j}$. On the other hand, Theorem
  \ref{theosuf} here does not put any constraints on $Z$ or $q_{j}$.
\end{remark}
\begin{remark}
  Under the improper uniform prior on $\beta$ and proper priors on
  variance components, \cite{mich:morr:2016} provide conditions for
  posterior propriety that does not put restrictions on the rank of $Z$. But,
  \cite{mich:morr:2016} assume the exponentiated norm bound
  condition. 
\end{remark}
\subsubsection{Sufficient conditions for posterior propriety for binary data} 
\label{subs:binary}
Although the Bernoulli distribution is a special case of the binomial
distribution, the need for analyzing binary data is ubiquitous. Thus,
in this section, we separately state the sufficient conditions for
posterior propriety for the important special case when $y_{i}=0$ or
$1$, $i=1,\dots,n$, the prior for $\beta$ is $\pi(\beta) \propto 1 $ and the prior for
$\tau_{j}$ is as in \eqref{eq:tauprior}.
\begin{corollary}
\label{corsuf}
Assume the following conditions are satisfied
\begin{enumerate}
\item $X$ is of full column rank, and there exists a positive vector $e>0$ such that $e^{\top} X^{*} = 0$ where $X^{*}$ is an $n \times p$ matrix with the $i$th row $c_{i} x_{i}^{\top}$ and $c_{i} = 1-2y_{i}$, $i=1,\dots,n$;
\item $a_{j} > p/2,\,b_{j} > 0$ for $j = 1,...,r$;
\item    $\E |\delta_\circ|^{p} < \infty$, where $\delta_\circ \sim F$.
\end{enumerate}
Then \eqref{eq:cy} holds for binary responses.
\end{corollary}
\begin{remark}
  In \pcite{chen2002necessary} Theorem $4.2$, the sufficient
  conditions for posterior propriety for binary data include that
  $(X, Z)$ is of full column rank, and there exists a positive vector $e$ such
  that $e^{\top}(X^*, Z^*)=0$. Here, $(X^*, Z^*)$ is defined in the same
  way as our $X^*$. Since, in practice, $Z$ is often a matrix
  consisting of $0$'s and $1$'s, the full column rank assumption of $Z$ may
  not hold. Furthermore, in \cite{wang2018convergence}, the sufficient
  conditions for posterior propriety for binary data include $W$,
  which is closely related to $(X, Z)$, is of full column rank, and there
  exists a positive vector $e$ such that $e^{\top}W^*=0$. Note that
  $W^*$ is defined similarly to our $X^*$. It also includes some
  conditions on $a_{j}$'s and $q_{j}$'s. In the above Corollary, we do
  not assume any conditions on the $Z$ matrix or $q_{j}$'s.
\end{remark}

The proof of the Corollary~\ref{corsuf} can be derived from the proof of Theorem \ref{theosuf} since $X^*_\triangle$ becomes $X^{*}$ and $k=0$ in the case of binary responses.  

\subsection{Sufficient conditions for posterior propriety for Poisson GLMMs} 
\label{subs:poisson} 
For Poisson GLMMs, since the log link is used, the likelihood function for $(\beta,\tau)$ is 
\begin{align*}
L(\beta,\tau \mid y)= \int_{\mathbb{R}^{q}}\prod_{i=1}^{n} \frac{\exp[-\exp(x_{i}^\top \beta + z_{i}^\top u)]\exp[(x_{i}^\top \beta + z_{i}^\top u)y_{i}]}{y_{i} !} \phi_{q}(u;0,D(\tau)^{-1}) du.
\end{align*}
Thus, if
\begin{align}
\label{eq:cy_poisson}
&c(y) = \int_{\mathbb{R}_{+}^{r}} \int_{\mathbb{R}^{p}} \int_{\mathbb{R}^{q}} \prod_{i=1}^{n} \frac{\exp[(x_{i}^\top \beta + z_{i}^\top u)y_{i}]}{\exp[\exp(x_{i}^\top \beta + z_{i}^\top u)] y_{i} !} \phi_{q}(u;0,D(\tau)^{-1}) du \pi(\tau)d\beta d\tau <\infty, 
\end{align}
the posterior density of $(\beta, \tau)$ is proper. In Lemma~\ref{suf_poisson}, we provide sufficient conditions for posterior propriety for Poisson data when the prior for $\beta$ is $\pi(\beta) \propto 1 $ and the prior for $\tau_{j}$ is as in \eqref{eq:tauprior}. Let $y_{(n)} =\max(y_{1},y_{2},\dots,y_{n})$. Letting $y_{(n)} = m_i, i= 1,\dots,n$,  we define $X^*_\triangle $ as that in Section \ref{subs:binomial}. 
\begin{lemma}
\label{suf_poisson}
Assume the following conditions are satisfied
\begin{enumerate}
\item $X$ is of full column rank, and there exists a positive vector $e>0$ such that $e^{\top} X^*_\triangle = 0$;
\item $a_{j} > p/2,\,b_{j} > 0$ for $j = 1,...,r$.
\end{enumerate}
Then \eqref{eq:cy_poisson} holds for Poisson GLMMs with the log link.
\end{lemma}
\begin{remark}
 In their sufficient conditions for
  posterior propriety for Poisson GLMMs, both \pcite{chen2002necessary} Theorem $3.1$ and  \pcite{natarajan2000reference} Corollary $1$ include $y_{i} >0$ for  at least $p$ observations. They also require the
  sub-matrix of $X$ corresponding to these observations to have full
  rank. These conditions do not hold for the Poisson GLMMs examples in
  Section~\ref{sec:application}, whereas the conditions of
  Lemma~\ref{suf_poisson} are
  satisfied. 
\end{remark}

\subsection{Sufficient conditions for posterior propriety for binomial and Poisson GLMMs with power priors} 
\label{subs:powerprior}
In this Section, we assume $b_{j}=0$ in \eqref{eq:tauprior}, that is,
$\tau_{j}$'s have the so-called power priors,
$\pi(\tau_j) \propto \tau_j^{a_j - 1}, j=1,\dots,r$. As mentioned in
\cite{roman2012convergence}, for the two-level normal model, the
standard diffuse prior on $\tau_j$ is
$\pi(\tau_j) \propto \tau_j^{-1/2 - 1}$, which is among the priors
recommended by \cite{gelm:2006}. For the prior on $\beta$, we continue
to assume $\pi(\beta) \propto 1$. Under these priors,
Theorem~\ref{theosuf_powerprior} provides sufficient conditions for
posterior propriety for binomial responses. Recall the notation
$I_3, t_i's$ and the matrix $X^*_\triangle$ as defined in Section
\ref{subs:binomial}. Also recall that $Z$ is the $n \times q$ random
effect design matrix with the $i$th row $z_{i}^\top$. Define
$\breve{Z}$ be the $k \times q$ matrix with rows $z_{i}^\top$ where
$i \in I_{3}$ and $Z_{\triangle} = (Z^\top:\breve{Z}^\top)^\top$ with
the $i$th row $z_{\triangle i}^\top$. Let $Z^*_{\triangle}$ be a
$(n+k) \times q$ matrix with the $i$th row as
$t_{i} z_{\triangle i}^\top$.
\begin{theorem}
\label{theosuf_powerprior}
Assume the following conditions are satisfied
\begin{enumerate}
\item $(X, Z)$ is of full column rank, and there exists a positive vector $e>0$ such that $e^{\top} (X^*_\triangle,Z^*_\triangle) = 0$;
\item $-q_{j}/2 < a_{j} < 0$ for $j = 1,...,r$;
\item    $\E |\delta_\circ|^{p-2 \sum_{j=1}^{r}  a_{j}} < \infty$, where $\delta_\circ \sim F$.
\end{enumerate}
Then \eqref{eq:cy} holds.
\end{theorem}
\begin{remark}
For the power prior on the variance components of the random effects in GLMMs, \cite{chen2002necessary} examine the sufficient conditions for posterior propriety for binary data in their Theorem $4.2$. Here, we derive sufficient conditions for posterior propriety for binomial data in GLMMs. The conditions in Theorem~\ref{theosuf_powerprior} are similar to those in \pcite{chen2002necessary} Theorem $4.2$ for binary data. 
\end{remark}
\begin{remark}
  In \pcite{wang2018convergence} Theorem $1$, the sufficient
  conditions for posterior propriety for binary data in GLMMs with the
  power prior on the variance components for the random effects have
  also been studied. Their conditions on $(a_j, q_j)$ are $a_{j} < 0$, $q_{j} \ge 2$
  and $a_{j} + q_{j}/2 > 1/2$ for $j=1,\dots, r$. On the other hand,
  the condition $2$ of Theorem~\ref{theosuf_powerprior} matches the
  hyperparameter condition in \cite{nata:mccu:1995} who consider
  single variance component Bernoulli
  GLMM with known $\beta$. 
\end{remark}
For binary data, let $Z^{*}$ be the $n \times q$ matrix with the $i$th
row $c_{i} z_{i}^{\top}$ and $c_{i} = 1-2y_{i}$, $i=1,\dots,n$. We thus
have the following Corollary.
\begin{corollary}
\label{cor_suf_binary_power}
Assume the following conditions are satisfied
\begin{enumerate}
\item $(X, Z)$ is of full column rank, and there exists a positive vector $e>0$ such that $e^{\top} (X^*,Z^*) = 0$;
\item $-q_{j}/2<a_{j} < 0$ for $j = 1,...,r$;
\item    $\E |\delta_\circ|^{p-2 \sum_{j=1}^{r}  a_{j}} < \infty$, where $\delta_\circ \sim F$.
\end{enumerate}
Then \eqref{eq:cy} holds for binary responses.
\end{corollary}

For Poisson GLMMs with the log link, as in Section \ref{subs:poisson}, we let $y_{(n)}= m_i$ for $i=1,\dots,n$  and define $(X^*_\triangle,Z^*_\triangle)$ following the notations in Section \ref{subs:binomial}. Also, using the relationship between the Poisson and binomial likelihoods as in the proof of Lemma~\ref{suf_poisson}, we have the following result.
\begin{lemma}
\label{cor_suf_poisson_power}
Assume the following conditions are satisfied
\begin{enumerate}
\item $(X, Z)$ is of full column rank, and there exists a positive vector $e>0$ such that $e^{\top} (X^*_\triangle,Z^*_\triangle) = 0$;
\item $-q_{j}/2 < a_{j} < 0$ for $j = 1,...,r$.
\end{enumerate}
Then \eqref{eq:cy_poisson} holds for Poisson GLMMs with the log link.
\end{lemma}
\section{Necessary conditions for posterior propriety for GLMMs}
\label{sec:nece}
\subsection{Necessary conditions for posterior propriety for binomial data}   
\label{subs:nec_binomial}
In this section, we discuss necessary conditions for posterior propriety for binomial data when the prior for $\beta$ is $\pi(\beta) \propto 1 $ and the prior for $\tau_{j}$ is as in \eqref{eq:tauprior} with $b_j \geq 0$ and the values of $a_j$ specified in the results of this section. 
\begin{theorem}
\label{theo_nece1}
For binomial data, the following two conditions are necessary for the posterior propriety, that is, for \eqref{eq:cy} to hold: 
\begin{enumerate}
\item $X$ is of full column rank;
\item $a_{j} +q_{j}/2 >0$ for $j = 1,...,r$.\end{enumerate}
\end{theorem}
If $Z$ is of full column rank, we have another necessary condition for posterior propriety in the case of binomial data when $\pi(\beta) \propto 1 $ and the prior for $\tau_{j}$ is as in \eqref{eq:tauprior}.
\begin{theorem}
  \label{theo_nece2}
  Suppose $Z$ has full column rank. The following are necessary conditions for posterior propriety, that is, for \eqref{eq:cy} to hold: 
\begin{enumerate}
\item $X$ is of full column rank;
\item There exists a positive vector $e>0$ such that $e^{\top} X^*_\triangle = 0$;
\item $a_{j} +q_{j}/2 >0$ for $j = 1,...,r$.
\end{enumerate}
\end{theorem}
\begin{remark}
\cite{sun2001propriety} study necessary conditions for posterior propriety for linear mixed models. In their Theorem $2$, the necessary conditions on the hyperparameters of the prior of the variance for the random effects include the condition $2$ of Theorem~\ref{theo_nece1} (condition $3$ of Theorem~\ref{theo_nece2}). 
In addition, we prove that $X$ full column rank is a necessary condition while \cite{sun2001propriety} derive their necessary conditions under the assumption that $X$ has full column rank. \pcite{chen2002necessary} Theorem $3.2$ also provides necessary conditions for posterior propriety for GLMMs. Their conditions are on the hyperparameter of the priors based on the spectral decomposition for the covariance matrix of the random effects. 
\end{remark}

\subsubsection{Necessary conditions for posterior propriety for binary data}   
\label{subs:nec_binary}
The necessary conditions for posterior propriety for binary data follow from Theorems \ref{theo_nece1} and \ref{theo_nece2}. 
\begin{corollary}
For binary data, the following are necessary conditions for posterior propriety, that is, \eqref{eq:cy}: 
\begin{enumerate}
\item $X$ is of full column rank;
\item $a_{j} +q_{j}/2 >0$ for $j = 1,...,r$.
\end{enumerate}
\end{corollary}
\begin{corollary} 
If $Z$ has full column rank, for binary data, the following are necessary conditions for posterior propriety, that is, \eqref{eq:cy}: 
\begin{enumerate}
\item $X$ is of full column rank;
\item There exists a positive vector $e>0$ such that $e^{\top} X^{*} = 0$;
\item $a_{j} +q_{j}/2 >0$ for $j = 1,...,r$.
\end{enumerate}
\end{corollary}
\begin{proof}
Since $X^*_\triangle$ becomes $X^{*}$ for binary responses, the proof can be derived from that for Theorem \ref{theo_nece2}.
\end{proof}
\begin{remark}
The necessary conditions for posterior propriety for Generalized linear models (GLMs) for binary responses have been explored in \cite{chen2001propriety}. The common necessary condition between their Theorem $2.2$ and our result is that $X$ is of full column rank. Further, under $Z$ full column rank, we derive the same necessary condition as their Theorem $2.2$. Note that their necessary conditions have been derived under an assumption on the link function. Furthermore, we establish another necessary condition on $a_{j}$ and $q_{j}$, which does not arise for GLMs. 
\end{remark}

Note that the full column rank of $X$ is a common necessary condition in all
results of this Section so far. Section~\ref{sub:necessary_general}
shows that the condition is required for propriety for all exponential
family GLMMs.

\subsection{Necessary conditions for posterior propriety for exponential family GLMMs}
\label{sub:necessary_general}
Suppose that the $i$th response variable $Y_{i}$ has the density
\begin{align}
\label{eq:exp}
p(y_{i} \mid \theta_{i}, \rho) = \exp[\rho (y_{i} \theta_{i} - b(\theta_{i})) + d(y_{i},\rho)],\quad i=1,\dots,n, 
\end{align}
where $\rho$ is the scalar dispersion parameter, $\theta_{i}$ is the
canonical parameter and the functions $b(\theta_i)$, $d(y_i, \rho)$
are known. Conditional on $\theta_{i}$ and $\rho$, assume that
$Y_{i}$'s are independent. Binomial, Poisson, and several other
popular distributions can be represented in the form of the
exponential family \eqref{eq:exp}
\citep{mccullagh2019generalized}. Here, we assume that $\rho$ is known
and is one, which is the case for binomial and Poisson families. In GLMMs, as discussed before, the canonical
parameter $\theta_{i}$ is related to $X$, $Z$, $\beta$ and $u$ by
$\theta_{i} = \theta(\eta_{i})$, where $\theta$ is a monotonic
differentiable function, referred to as the $\theta$-link. Note that,
$\E(Y_{i}) = b'(\theta_{i})$. When $\theta_{i} = \eta_{i}$, the link
is said to be canonical. Assuming $u \sim N(0, \Psi)$ where $\Psi$ is
a positive definite matrix, the likelihood function of $(\beta, \Psi)$
is
\begin{align}
\label{eq:general_expo_likef}
L(\beta, \Psi \mid y)= \int_{\mathbb{R}^{q}}\bigg[\prod_{i=1}^{n} \exp[ (y_{i} \theta_{i} - b(\theta_{i})) + d(y_{i})]\bigg] \phi_{q}(u;0,\Psi) du,
\end{align}
where $d(y_{i}) = d(y_{i},\rho)$. As before, we
consider $\pi(\beta) \propto 1 $ and the
prior on $\Psi$ is denoted by $\pi(\Psi)$. We assume $\beta$ and
$\Psi$ are apriori independent. The corresponding posterior density of
$(\beta, \Psi)$ is $\pi(\beta,\Psi \mid y) \propto L(\beta,\Psi \mid y)\pi(\Psi)$.
Now, the joint posterior density of $(\beta, \Psi)$ is proper, if the marginal density of
$y$
\begin{align}
\label{eq:cy_expo}
c(y) = \int_{\mathfrak{A}} \int_{\mathbb{R}^{p}}
\int_{\mathbb{R}^{q}}\bigg[\prod_{i=1}^{n} \exp[ (y_{i} \theta_{i} - b(\theta_{i})) + d(y_{i})]\bigg] \phi_{q}(u;0,\Psi) du \pi(\Psi)
d\beta d\Psi < \infty,
\end{align}
where $\mathfrak{A}$ is an appropriate space of positive definite matrices where $\Psi$ lies.

We have the following necessary conditions for the propriety of $\pi(\beta,\Psi \mid y)$. 
\begin{theorem}
\label{theonec_expo}
Assume $\pi(\Psi)$ is a proper prior for $\Psi$ and $b(\theta)$ in \eqref{eq:exp} is a monotone function, then rank $(X) =p$ is a necessary condition for the propriety of the posterior density of $(\beta,\Psi)$, that is \eqref{eq:cy_expo}, in exponential family GLMMs. 
\end{theorem}
\begin{remark}
The function $b(\theta)$ is monotone for binomial, Poisson, gamma, and inverse Gaussian families. On the other hand, this condition does not hold for the normal distribution. 
\end{remark}
\begin{remark}
  The necessary condition of rank $(X) = p$ 
  is proved for GLMs for binary data in \pcite{chen2001propriety} Theorem $2.2$. In \pcite{sun2001propriety} Theorem $2$, the assumption of $X$ full column rank is included for establishing the necessary conditions for posterior propriety for linear mixed models.  
\end{remark}

\section{An approximate Jeffreys' Prior for GLMMs with canonical link}
\label{sec:jp}
Jeffreys' prior \citep{jeffreys1946invariant} is a reference prior
widely used for objective Bayesian inference. In this section, we
construct an approximate `independence Jeffreys' prior'
\citep{berger2001objective} for GLMMs. For ease of presentation, we
assume that $\Psi = D(\tau)^{-1} = \oplus_{j=1}^{r} (1/\tau_{j}) \I_{q_{j}}$, as in
Section~\ref{sec:suffi}. Recall that the likelihood function for $(\beta,\tau)$,
$L(\beta,\tau \mid y)$ follows from \eqref{eq:general_expo_likef}.  Define
the complete likelihood function
\begin{align}
\label{eq:L_nointegral}
  L(\beta,\tau \mid y, u) &= \prod_{i=1}^{n} p(y_{i} \mid \beta,u) \phi_{q}(u;0,D(\tau)^{-1}) \nonumber\\
  & =  \prod_{i=1}^{n} p(y_{i} \mid \beta,u) (2\pi)^{-q/2} \prod_{j=1}^{r} \tau_{j}^{q_{j}/2} \exp(-\tau_{j} u_{j}^{\top} u_{j}/2).
\end{align}
From \cite{casella2001empirical} we have 
 \begin{align}
 \label{eq:second_deri_l}
 \frac{d^2 }{d \tau_{j}^2}\log L(\beta, \tau \mid y)= \E \Big[\frac{d^2}{d \tau_{j}^2} \log L(\beta, \tau \mid y, u) \Big| \tau_{j}, y \Big]+ \Var \Big[\frac{d}{d \tau_{j}} \log L(\beta, \tau \mid y, u) \Big| \tau_{j}, y \Big], 
 \end{align}
where the expectation and variance are with respect to $\pi(u_{j} \mid \tau_{j},y)$. Next, using \eqref{eq:L_nointegral}, we obtain
\begin{align}
 \label{eq:first_deri_lyu}
\frac{d}{d \tau_{j}} \log L(\beta,\tau \mid y, u)=q_{j} \tau_{j}^{-1}/2 - u_{j}^\top u_{j}/2,\;\mbox{and}\;  \frac{d^2}{d \tau_{j}^2} \log L(\beta,\tau \mid y, u)=-q_{j} \tau_{j}^{-2}/2.
\end{align}
Then, by using \eqref{eq:first_deri_lyu}, we get the expectation in \eqref{eq:second_deri_l} as
\begin{align}
\label{eq:expect_derived}
\E \Big[\frac{d^2}{d \tau_{j}^2} \log L(\beta, \tau \mid y, u) \Big| \tau_{j}, y \Big] =\E(-q_{j} \tau_{j}^{-2}/2 \mid \tau_{j}, y)=-q_{j} \tau_{j}^{-2}/2. 
\end{align}
Also, by using \eqref{eq:first_deri_lyu}, we have the variance in \eqref{eq:second_deri_l} as
\begin{align}
\label{eq:var_derived}
\Var \Big[\frac{d}{d \tau_{j}} \log L(\beta, \tau \mid y, u) \Big| \tau_{j}, y \Big] = \Var(u_{j}^\top u_{j} \mid \tau_{j}, y)/4 \approx \sum_{m=1}^{q_{j}} \Var (u_{jm}^2 \mid \tau_{j}, y)/4, 
\end{align}
Here, we have ignored the covariance terms. 

Let $y_{1.} = (y_{11},\dots, y_{1n_{1}}) $ denote the observations for
the level $1$ of the random effect $u_{j}$. The conditional distribution of
$u_{j1} \mid y_{1.}, \beta, \tau_{j}$ is
\begin{align}
f(u_{j1} \mid y_{1.},\beta, \tau_{j}) = \frac{\prod_{k=1}^{n_{1}} p(y_{1k} \mid \beta, u_{j1}) \phi(u_{j1}; 0, 1/\tau_{j})}{\int_{\mathbb{R}}\prod_{k=1}^{n_{1}} p(y_{1k} \mid \beta, u_{j1}) \phi(u_{j1}; 0, 1/\tau_{j}) d u_{j1}}, 
\end{align}
where $p(y_{1k} \mid \beta, u_{j1}) $ is the exponential density \eqref{eq:exp} with $\rho=1$, and $\phi(u_{j1}; 0, 1/\tau_{j})$ is the normal density for $u_{j1}$ with mean $0$ and variance $1/\tau_{j}$. 

\cite{booth1998standard} derive some approximations to the conditional mean and variance of the random effects. Following \cite{booth1998standard}, let 
\begin{align*}
l &= \log \prod_{k=1}^{n_{1}} p(y_{1k} \mid \beta, u_{j1}) \phi(u_{j1}; 0, 1/\tau_{j})\\& = \sum_{k=1}^{n_{1}}(y_{1k} \theta_{1k} - b(\theta_{1k}) + d(y_{1k}))-u_{j1}^2 \tau_{j}/2 - \log(2 \pi)/2 + \log \tau_{j}/2.
\end{align*}  
Since $\theta_{1k} = \eta_{1k}$, we have
\begin{align*}
l^{(1)}=\frac{d}{d u_{j1}} l =  \sum_{k=1}^{n_{1}} (y_{1k} - b'(\theta_{1k})) z_{j1k} - \tau_{j} u_{j1},
\end{align*}
and 
\begin{align*}
l^{(2)}=\frac{d^2}{d u_{j1}^2} l = - \tau_{j} - \sum_{k=1}^{n_{1}} z_{j1k}^2 \;t'(x_{1k}^\top \beta + z_{1k}^\top u).
\end{align*}
Here $z_{j1k}$ is the corresponding value from the $Z$ matrix. Note that $b'(\theta_{1k}) = \E (y_{1k} ) = t(x_{1k}^\top \beta + z_{1k}^\top u)$, where $t(\eta_{ik})$ is a function of $\eta_{ik}$, and $t'(x_{1k}^\top \beta + z_{1k}^\top u) =t'(\eta_{ik})$ is the derivative with respect to $\eta_{ik}$. Letting $l^{(1)}=0$, we get 
\[\tilde{u}_{j1}= \E(u_{j1} \mid \tau_{j},y_{1.})= \frac{\sum_{k=1}^{n_{1}} z_{j1k}(y_{1k} - \E(y_{1k}))}{\tau_{j}}.
\] Furthermore, as in \cite{booth1998standard}, we assume
$ \Var(u_{j1} \mid \tau_{j},y_{1.}) \approx - (l^{(2)})^{-1}$. Also,
in the above expression for the variance, we replace
$t'(x_{1k}^\top \beta + z_{1k}^\top u)$ with
$t'(x_{1k}^\top \hat{\beta})$, where $\hat{\beta}$ is the MLE of the
regression coefficients from the GLM without the random effects. Then,
it follows that
\[\tilde{v}_{j1}= \Var(u_{j1} \mid \tau_{j},y_{1.})\approx \Big(\tau_{j} + \sum_{k=1}^{n_{1}} z_{j1k}^2 \;t'(x_{1k}^\top \hat{\beta})\Big)^{-1}
\]

Next, we assume the conditional distribution of $u_{j1} \mid \tau_{j}, y_{1.}$ is the normal distribution with mean $\tilde{u}_{j1}$ and variance $\tilde{v}_{j1} $. Then we have $\Var(u_{j1}^2 \mid \tau_{j},y_{1.})=2 \tilde{v}_{j1}^2 + 4 \tilde{v}_{j1} \tilde{u}_{j1}^2.$ Repeating the above method for $u_{jm}, \, m=2,\dots,q_{j}$, \eqref{eq:var_derived} becomes
\begin{align}
\label{eq:var_derived_new}
\Var \Big[\frac{d}{d \tau_{j}} \log L(\beta, \tau \mid y, u) \Big| \tau_{j} \Big]  &\approx \sum_{m=1}^{q_{j}} [ \tilde{v}_{jm}^2/2 +  \tilde{v}_{jm} \tilde{u}_{jm}^2] \nonumber\\
&\approx \sum_{m=1}^{q_{j}}  \tilde{v}_{jm}^2/2 = \sum_{m=1}^{q_{j}} \big(\tau_{j} + \sum_{k=1}^{n_{m}} z_{jmk}^2 \;t'(x_{mk}^\top \hat{\beta})\big)^{-2}/2,
\end{align}
where the second approximation is using $\tilde{u}_{jm} \approx 0$.

Finally, by using \eqref{eq:expect_derived},
\eqref{eq:var_derived_new} in \eqref{eq:second_deri_l}, we obtain the
expected Fisher information for $\tau_{j}$ as
$q_{j} \tau_{j}^{-2}/2 -\sum_{m=1}^{q_{j}} \big(\tau_{j} +
\sum_{k=1}^{n_{m}} z_{jmk}^2 \;t'(x_{mk}^\top
\hat{\beta})\big)^{-2}/2.$ Thus, the approximate Jeffreys' prior for $\tau_i$ is
\begin{equation}
  \label{eq:jeff}
\pi_J(\tau_{i}) = \big[q_{i} \tau_{i}^{-2}/2 -\sum_{m=1}^{q_{i}} \big(\tau_{i} + \sum_{k=1}^{n_{m}} z_{imk}^2 \;t'(x_{mk}^\top \hat{\beta})\big)^{-2}/2 \big]^{1/2}, \;\text{for}\; i=1,\dots,r.   
\end{equation}
Finally, the `independence Jeffreys' prior
\citep{berger2001objective,kass1996selection} for $(\beta, \tau)$ is
$\pi_J(\beta, \tau) = \prod_{i=1}^r \pi_J(\tau_{i})$.

\subsection{Posterior Propriety for binomial and Poisson GLMMs}
For binomial and Poisson families, $t'(x_{mk}^\top \hat{\beta})$ in
\eqref{eq:jeff} is positive. Denoting
$\sum_{k=1}^{n_{m}} z_{imk}^2 \;t'(x_{mk}^\top \hat{\beta})$ by $c_{im} \ge 0$,
from \eqref{eq:jeff}, we have 
\begin{align}
  \label{eq:jeffprop}
  \pi_J(\tau_{i}) = \big[q_{i} \tau_{i}^{-2}/2 -\sum_{m=1}^{q_{i}} (\tau_{i} + c_{im})^{-2}/2 \big]^{1/2}
   = \bigg[\sum_{m=1}^{q_i} \Big\{\frac{1}{2\tau_i^{2}} -  \frac{1}{2(\tau_i +c_{im})^{2}}\Big\} \bigg]^{1/2}.
\end{align}
Now,
\begin{align}
  \label{upper}
  \frac{1}{2\tau_i^{2}} -  \frac{1}{2(\tau_i +c_{im})^{2}} &= \frac{c_{im}^2 + \tau_i c_{im} + \tau_i c_{im}}{2 \tau_i^2 (\tau_i + c_{im})^2}  =  \frac{c_{im} (\tau_i + c_{im}) + \tau_i c_{im}}{2 \tau_i^2 (\tau_i + c_{im})^2} \nonumber\\
&= \frac{c_{im}}{2\tau_i^2 (\tau_i+c_{im})} + \frac{c_{im}}{2\tau_i(\tau_i+c_{im})^2} \le  \frac{\sqrt{c_{im}}}{4 \tau_i^{5/2}} + \frac{\sqrt{c_{im}}}{4 \tau_i^{5/2}  } = \frac{\sqrt{c_{im}}}{2 \tau_i^{5/2}  },
\end{align}
where we use $\tau_i + c_{im} \ge 2 \sqrt{\tau_i c_{im}}$ and
$(\tau_i + c_{im})^2 \ge 2 \tau_i\sqrt{\tau_i c_{im}}$ for the
inequality. Thus, from \eqref{eq:jeffprop} and \eqref{upper}, we have
\begin{equation}
  \label{eq:jeffup}
  \pi_J(\tau_{i}) \le \Big(\sum_{m=1}^{q_i} \sqrt{c_{im}}/2\Big)^{1/2} \tau_i^{-5/4}.  
\end{equation}
Since the upper bound in \eqref{eq:jeffup} is a power prior with
$a_i = -1/4$, sufficient conditions for posterior propriety with the
proposed Jeffreys' prior $\pi_J (\beta, \tau)$ for binomial and
Poisson GLMMs follows from Theorem~\ref{theosuf_powerprior} and
Corollary~\ref{cor_suf_poisson_power}, respectively. Thus, if the
condition $1$ of Theorem~\ref{theosuf_powerprior} holds then the
posterior densities of binomial and
Poisson GLMMs with canonical links and the prior $\pi_J (\beta, \tau)$ are proper.
\begin{remark}
  \cite{natarajan2000reference} derive conditions for posterior
  propriety for GLMMs with an approximate Jeffreys' prior (see
  Section~\ref{sec:natajeff} for some details on this prior). In
  general, \pcite{natarajan2000reference} conditions require a complex,
  multi-dimensional integral to be finite. For Poisson GLMMs with the
  log link, their conditions require $y_{i} >0$ for $p$
  observations, the corresponding sub-matrix of $X$ is full rank, and
  for the binary GLMM, \cite{natarajan2000reference} assume some
  intractable integrals to be finite, along with other conditions on
  $(X, y)$, whereas the sufficient condition given here for posterior
  propriety with the proposed Jeffreys'
  prior 
  can be easily verified \citep{roy2007convergence}.
\end{remark}
\subsection{Comparison with \pcite{natarajan2000reference} prior}
\label{sec:natajeff}
\cite{natarajan2000reference} provide an approximate Jeffreys' prior
for $\Psi$ for GLMMs. The general form of \pcite{natarajan2000reference} prior is
complex involving sums of traces of inverse of certain matrices and
derivatives of some matrices, although closed-form expressions can be
derived for binary and Poisson one-way random intercept GLMMs. For example, if
$p= q=q_{1}=r=1$, \pcite{natarajan2000reference} prior for binary data is 
\begin{align*}
\pi_{NK}(\tau) \propto \Big(1+n e^{\hat{\beta} } \tau/(1+e^{\hat{\beta}} )^2 \Big)^{-1},
\end{align*}
and for Poisson data, is
\begin{align*}
\pi_{NK}(\tau) \propto \big(1+n e^{\hat{\beta} } \tau\big)^{-1},
\end{align*}
where $n$ is the number of observations.

Note that for binary data,
$d(1+e^{-\eta})^{-1}/d \eta = (1+e^{-\eta})^{-2} e^{-\eta}$, thus,
$t'(x_{mk}^\top \hat{\beta})= (1+e^{-\hat{\beta}})^{-2}
e^{-\hat{\beta}}$, and for the Poisson GLMMs,
$d e^{\eta}/d \eta = e^{\eta}$, thus,
$t'(x_{mk}^\top \hat{\beta})= e^{\hat{\beta}}$. Hence, from
\eqref{eq:jeffprop}, if
$p= q=q_{1}=r=1$, for binary responses we have
\begin{align*}
\pi_{J}(\tau) \propto \big[\tau^{-2}/2 - \big(\tau +ne^{\hat{\beta}}/(1+e^{\hat{\beta}})^2\big)^{-2}/2 \big]^{1/2}, 
\end{align*}
and for the Poisson GLMMs
\begin{align*}
\pi_{J}(\tau) \propto \big[\tau^{-2}/2 - \big(\tau +ne^{\hat{\beta}}\big)^{-2}/2 \big]^{1/2}. 
\end{align*}
In the Appendix, we compare the behavior of tails of $\pi_{J}(\tau)$
and $\pi_{NK}(\tau)$, and we observe that practically $\pi_J$ can be more
diffuse than $\pi_{NK}$.

\section{Examples}
\label{sec:application}
In this section, we use two popular examples, namely, the one-way
random effects model and the two-way random effects model to
demonstrate how our theoretical results can be easily applied to verify
posterior propriety for GLMMs.

\subsection{An example involving one-way random effects models}
\label{sec:oneway}
The data model is
$ Y_{i} \mid \beta,u,\tau \overset{ind}{\sim} Binomial(m_{i},
F(x_{i}^\top \beta + z_{i}^\top u))$ for $i = 1,\dots,n$ with
$u \mid \tau \sim N(0, [1/\tau] I).$ Here, we consider $p=2$ and $r=1$
random effect with $q = q_{1} = 2$. Also, we consider $n=6$, and
$y=(0,\,4,\,2,\,4,\,3,\,5)$ as the observed binomial responses. Let
$m_{1}=3,\, m_{2}=4,\,m_{3}=5,\,m_{4}=4,\,m_{5}=3$ and $m_{6}=5$.  The
design matrix $X$ and the random effect matrix $Z$ are given by
\[
X^\top=\begin{pmatrix}
1 &1&1&1&1&1\\
2.9&1.7&2.6&3.1&3.8&4.2
\end{pmatrix},\;
\text{and}\;
Z^\top=\begin{pmatrix}
1 &1&1&0&0&0\\
0&0&0&1&1&1
\end{pmatrix} 
\]
Then, based on the notations in Section \ref{subs:binomial}, we obtain
\[
X_{\triangle}^\top=\begin{pmatrix}
1&1&1&1&1&1&1\\
2.9&1.7&2.6&3.1&3.8&4.2&2.6
\end{pmatrix},\;
\text{and}\]
\[
{X^*_\triangle}^{\top}=\begin{pmatrix}
1&-1&1&-1&-1&-1&-1\\
2.9&-1.7&2.6&-3.1&-3.8&-4.2&-2.6
\end{pmatrix}
\]
Note that $X$ is of full column rank. To apply Theorem~\ref{theosuf}, we need
to check if there exists a positive vector $e>0$ such that
$e^{\top} X^*_\triangle = 0$. Using the function `simplex' in the R
package `boot', we easily find that the above condition is satisfied (see
\pcite{roy2007convergence} Appendix A for details). That is, the condition $1$ in
Theorem~\ref{theosuf} holds. We can choose the hyperparamters for the
prior of $\tau$ satisfying the condition $2$ in Theorem
\ref{theosuf}. If we consider the probit or the logit link, the
condition $3$ in Theorem \ref{theosuf} is satisfied. Therefore, the
resulting posterior densities of $(\beta, \tau)$ are proper. On the other hand,
since many observed $y_{i}$'s are zero as well as $m_{i}$, results in
\cite{chen2002necessary} cannot be used to establish posterior
propriety here. In particular, as $0< y_i <m_i$ only for $i=3$, there
does not exist a full column rank sub-matrix $X_s$ of $X$ with
$0< y_{s_i} <m_{s_i}, i=1,\dots,p$. Thus, \pcite{chen2002necessary}
results are not applicable.

We also consider the Poisson GLMM, where
$ Y_{i} \mid \beta,u,\tau \overset{ind}{\sim} Poisson(\exp[x_{i}^\top
\beta + z_{i}^\top u])$ for $i = 1,\dots,n$ with
$u\mid \tau \sim N(0, [1/\tau] I).$ We consider $p=2$ and $r=1$ random
effect with $q=q_{1}=2$. Also, we have $n=6$ and
$y=(0,\,0,\,0,\,2,\,0,\,0)$ as the observed Poisson responses. The
design matrix $X$ and the random effect matrix $Z$ are given by
\[
X^\top=\begin{pmatrix}
1 &1&1&1&1&1\\
9.4&8.7&10.2&9.1&8.9&9.5
\end{pmatrix},\;
\text{and}\;
Z^\top=\begin{pmatrix}
1 &1&1&0&0&0\\
0&0&0&1&1&1
\end{pmatrix} 
\]
Then, based on the notations in Section \ref{subs:binomial} and \ref{subs:poisson}, we obtain
\[
X_{\triangle}^\top=\begin{pmatrix}
1&1&1&1&1&1\\
9.4&8.7&10.2&9.1&8.9&9.5
\end{pmatrix},\;
\text{and}\]
\[
{X^*_\triangle}^{\top}=\begin{pmatrix}
1&1&1&-1&1&1\\
9.4&8.7&10.2&-9.1&8.9&9.5
\end{pmatrix}
\]
Here, $X$ has full column rank. Using similar methods as the last example, the conditions in Lemma~\ref{suf_poisson} hold. Hence, the propriety of the posterior of $(\beta, \tau)$ is obtained. Since many observed $y_{i}$'s are zero, existing results in \cite{chen2002necessary} and \cite{natarajan2000reference} cannot be used to establish posterior propriety for this example. In particular, as $0< y_i $ only for $i=4$, there
does not exist a full column rank sub-matrix $X_s$ of $X$ with
$0< y_{s_i} , i=1,\dots,p$. Thus, \pcite{chen2002necessary} and \pcite{natarajan2000reference}
results are not applicable.

\subsection{An example involving two-way random effects models}
Consider the binomial GLMM is
$ Y_{i} \mid \beta,u,\tau \overset{ind}{\sim} Binomial(m_{i},
F(x_{i}^\top \beta + z_{i}^\top u))$ for $i = 1,\dots,n$, with
$u =(u_1, u_2)^{\top}$ and
$ u_{j} \mid \tau_{j} \overset{ind}{\sim} N(0, [1/\tau_{j}]
\I_{q_{j}}), j = 1, 2.$ Suppose we have $r=2$ random effects,
$p=2,\,q_{1} = 3, \, q_{2} = 2$ and $q=5$. Also, we consider
$y=(0,\,1,\,2,\,0,\,2,\,2)$ as the observed binomial responses. Let
$m_{1}=m_{2}=m_{3}=m_{4}=m_{5}=m_{6}=2$. The design matrix $X$ and the
random effect matrix $Z$ are given by
\[
X^\top=\begin{pmatrix}
1&1&1&1&1&1\\
1.8&2.1&3.2&4.9&5.3&6.1
\end{pmatrix} \; \text{and}\;
Z^\top=\begin{pmatrix}
1&1 &0&0&0&0\\
0&0&1&1&0&0\\
0&0&0&0&1&1\\
1&0&1&0&1&0\\
0&1&0&1&0&1
\end{pmatrix}
\]
Then, based on the notations in Section \ref{subs:binomial}, we obtain
\[
X_{\triangle}^\top=\begin{pmatrix}
1&1&1&1&1&1&1\\
1.8&2.1&3.2&4.9&5.3&6.1&2.1
\end{pmatrix}, \;\text{and}
\] 
\[
{X^*_\triangle}^\top=\begin{pmatrix}
1&1&-1&1&-1&-1&-1\\
1.8&2.1&-3.2&4.9&-5.3&-6.1&-2.1
\end{pmatrix}
\] 
In this example, $X$ is a full column rank matrix while $Z$ is not. By
employing the same method as that in Section~\ref{sec:oneway}, we find
that there exists a positive vector $e>0$ such that
$e^{\top} X^*_\triangle = 0$. That is, the condition $1$ of
Theorem~\ref{theosuf} holds. We can choose the hyperparameters for the
priors of $\tau_{1}$ and $\tau_{2}$ satisfying the condition $2$ in
Theorem \ref{theosuf}. For probit or logistic GLMMs, the condition $3$
in Theorem \ref{theosuf} is satisfied. Hence, the resulting posterior densities of
$(\beta, \tau)$ are proper according to Theorem \ref{theosuf}. This
example demonstrates that $Z$ full column rank is not a necessary condition
for posterior propriety for GLMMs with the improper uniform prior on
$\beta$.  Also, since $0< y_i <m_i$ only for $i=2$, there does not
exist a full column rank sub-matrix $X_s$ of $X$ with
$0< y_{s_i} <m_{s_i}, i=1,\dots,p$. Thus, although
\pcite{chen2002necessary} results do not apply to this two-way random
effects example, their conditions on $(y, X)$ for binomial GLMMs do not hold.

One can analyze datasets arising in diverse disciplines  by fitting a
two-way random effects GLMM. For example, from the popular R package
`lme4' \citep{r:lme4}, consider the `grouseticks' dataset, which can
be analyzed by fitting a Poisson GLMM for the response
variable `ticks' with `year' and `height' as the
fixed effects, and `brood' and `location' as the random effects. Then,
from Lemma~\ref{suf_poisson}, propriety of the posterior of
$(\beta, \tau)$ follows when the improper uniform prior is used for
$\beta$ and proper gamma priors satisfying the condition $2$ in
Lemma~\ref{suf_poisson} are placed on $\tau$. Here,
$\beta =(\beta_0, \beta_1,\beta_2, \beta_3)$ with the intercept term
$\beta_0$, the fixed effects parameters $\beta_1, \beta_2$ of the
levels 96 and 97, respectively, of the variable year and the
regression coefficient $\beta_3$ of the continuous variable
height. Also, $\tau = (\tau_1, \tau_2)$ with $1/\tau_1 (1/\tau_2)$
being the variance of the random effect brood (location). 
Finally, since this example involves two random effects, existing
results in \cite{chen2002necessary} and \cite{natarajan2000reference},
who consider models for longitudinal data and cluster data, respectively, are not readily
applicable.

\section{Discussion}
\label{sec:disc}
We have derived the necessary and sufficient conditions for posterior
propriety for GLMMs under various widely used reference priors,
including the Jeffreys' prior. Unlike the available results in the
literature, the conditions presented here for binomial and Poisson
GLMMs can be easily verified. For example, our results do not assume
the strong exponentiated norm bound condition of \cite{mich:morr:2016}
or an intractable, multi-dimensional integral as in
\cite{natarajan2000reference} to be finite. Also, some existing
results on posterior propriety for GLMMs put some constraints on the
random effects design matrix $Z$ that may not hold in practice.  On
the other hand, some of our sufficient conditions for posterior
propriety do not put any constraint on $Z$, although they assume
proper priors on the variance components of the random effects. We
also provide sufficient conditions for posterior propriety when the
improper power prior or an approximate Jeffreys' prior is used on the
variance parameters of the random effects. Exploiting a relationship
between the likelihoods of the Poisson GLMMs with the log link and the
binomial GLMMs with the logit link, this article derives easily
verifiable conditions for posterior propriety for Poisson GLMMs with
the log link. In Section \ref{sec:suffi} for priors on the $r$ random
effects, we have considered the practical settings where either
$b_j = 0$ for all $j =1,\dots,r$ or $b_j > 0$ for all $j =1,\dots,r$,
although it might be of theoretical interest to study posterior
propriety in the situations where $b_j =0$ for some $j$ while it is
strictly positive for the other random effects. Finally, in many
modern datasets, $p$ is often larger than $n$. Thus, a potential
future work is to develop posterior propriety conditions for GLMMs in
the case of $p>n$.

\bigskip
\noindent {\Large \bf Appendix: Proofs of theoretical results}
\medskip
\begin{appendix}

\begin{proof}[Proof of Theorem~\ref{theosuf}]
  As in the proof of Theorem 1 of \cite{roy2013posterior}, we have
  \begin{align*}
    &    \binom{m_{i}}{y_{i}} \big{[}F(x_{i}^\top \beta + z_{i}^\top u)\big{]}^{y_{i}}\big{[}1- F(x_{i}^\top \beta + z_{i}^\top u) \big{]}^{m_{i}-y_{i}}\\
    & \le \left\{
\begin{array}{ll}
1- F(x_{i}^\top\beta + z_{i}^\top u) &\mbox{if} \,i \in I_{1}\\
F(x_{i}^\top\beta + z_{i}^\top u)  &\mbox{if} \,i \in I_{2}\\
{m_{i} \choose y_{i}} F(x_{i}^\top\beta + z_{i}^\top u) \big[1- F(x_{i}^\top\beta + z_{i}^\top u)\big]  &\mbox{if}\,i \in I_{3}.
\end{array}
\right.
  \end{align*}
Thus from \eqref{eq:cy} we have
\begin{align}
\label{eq:binosuf}
c(y) &\le \int_{\mathbb{R}_{+}^{r}} \int_{\mathbb{R}^{p}} \int_{\mathbb{R}^{q}} \bigg\{\prod_{i \in I_{1}} \big[1- F(x_{i}^\top\beta + z_{i}^\top u) \big] \bigg\}\bigg\{ \prod_{i \in I_{2}} F(x_{i}^\top\beta + z_{i}^\top u)\bigg\}  \nonumber\\
&\quad \bigg\{\prod_{i \in I_{3}} {m_{i} \choose y_{i}} F(x_{i}^\top\beta + z_{i}^\top u) \big[1- F(x_{i}^\top\beta + z_{i}^\top u) \big]\bigg\} \phi_{q}(u;0,D(\tau)^{-1}) du \pi(\tau)d\beta d\tau. 
\end{align}
If the random variable $\xi \sim F(\cdot)$, then
$1-F(x) = \E \,\I(\xi >x)$ and $F(x) = \E \, \I(-\xi \ge -x)$. Let
$\delta_{1}, \delta_{2}, \dots,\delta_{n+k} \stackrel{iid}{\sim} F(\cdot)$. Let
$\delta = (\delta_{1}, \delta_{2},\dots,\delta_{n+k})^\top$ and
$\delta^* = (t_{1}\delta_{1}, t_{2}\delta_{2},\dots,t_{n+k}
\delta_{n+k})^\top$, where $t_{i}$ is as defined before. Thus, using
\eqref{eq:tauprior}, the inequality \eqref{eq:binosuf} becomes
\begin{align}
c(y) &\le (2\pi)^{-\frac{q}{2}}\int_{\mathbb{R}_{+}^{r}} \int_{\mathbb{R}^{q}}  \bigg[\prod_{i \in I_{3}} {m_{i} \choose y_{i}}\bigg] \int_{\mathbb{R}^{p}} \E\,\big[\I\{t_{i} (x_{i}^\top\beta + z_{i}^\top u) \le t_{i} \delta_{i};1 \le i \le n+k \} \big]d\beta  \nonumber\\  &\quad\prod_{j=1}^{r} \tau_{j}^{a_{j}+q_{j}/2-1} \exp[-\tau_{j}(b_{j} + u_{j}^\top u_{j}/2)] d u d \tau  \label{eq:cy_binomial_ineq}\\
&= (2\pi)^{-\frac{q}{2}} \prod_{i \in I_{3}} {m_{i} \choose y_{i}} \int_{\mathbb{R}^{q}} \prod_{j=1}^{r} \frac{\Gamma{(a_{j}+q_{j}/2)}}{(b_{j}+ u_{j}^\top u_{j}/2)^{a_{j}+q_{j}/2}}\nonumber\\  &\quad \E \big[ \int_{\mathbb{R}^{p}} \I\{t_{i} (x_{i}^\top\beta + z_{i}^\top u) \le t_{i} \delta_{i};1 \le i \le n+k \} d\beta  \big] du \nonumber \\
&= (2\pi)^{-\frac{q}{2}} \prod_{i \in I_{3}} {m_{i} \choose y_{i}} \int_{\mathbb{R}^{q}} \prod_{j=1}^{r} \frac{\Gamma{(a_{j}+q_{j}/2)}}{(b_{j}+ u_{j}^\top u_{j}/2)^{a_{j}+q_{j}/2}}\nonumber\\  &\quad \E \big[\int_{\mathbb{R}^{p}} \I\{t_{i} x_{i}^\top\beta  \le t_{i} \delta_{i} -t_{i}z_{i}^\top u;1 \le i \le n+k \} d\beta  \big] du \nonumber\\
&= (2\pi)^{-\frac{q}{2}}\prod_{i \in I_{3}} {m_{i} \choose y_{i}} \int_{\mathbb{R}^{q}} \prod_{j=1}^{r} \frac{\Gamma{(a_{j}+q_{j}/2)}}{(b_{j}+ u_{j}^\top u_{j}/2)^{a_{j}+q_{j}/2}} \E \big[ \int_{\mathbb{R}^{p}} \I\{X^*_\triangle \beta  \le \delta^* -Z^*_{\triangle} u\} d\beta \big] du \nonumber\\
&\le (2\pi)^{-\frac{q}{2}} \prod_{i \in I_{3}} {m_{i} \choose y_{i}} \int_{\mathbb{R}^{q}} \prod_{j=1}^{r} \frac{\Gamma{(a_{j}+q_{j}/2)}}{(b_{j}+ u_{j}^\top u_{j}/2)^{a_{j}+q_{j}/2}}\nonumber\\  &\quad \E \big[\int_{\mathbb{R}^{p}} \I\big\{\norm{\beta}  \le l \norm{\delta^* -Z^*_{\triangle} u} \big\} d\beta \big]  du, \label{eq:cy_binomial}
\end{align} 
where the first equality follows from Tonelli's Theorem and the condition $2$ in Theorem~\ref{theosuf}. Note that $X$ is of full column rank by the condition $1$. Then $X_{\triangle}$ is also of full column rank. Also, there exists a positive vector $e>0$ such that $e^{\top} X^*_\triangle = 0$. Thus Lemma $4.1$ in \cite{chen2001propriety} can be used to get the last inequality, where $l$ is a constant depending on $X^*_\triangle$. Since $\norm{\delta^* -Z^*_{\triangle} u} \le \norm{\delta^*} + \norm{Z^*_{\triangle} u} $, from \eqref{eq:cy_binomial}, we have
\begin{align}
\label{eq:binosuff}
c(y) &\le (2\pi)^{-\frac{q}{2}}\prod_{i \in I_{3}} {m_{i} \choose y_{i}} \int_{\mathbb{R}^{q}} \prod_{j=1}^{r} \frac{\Gamma{(a_{j}+q_{j}/2)}}{(b_{j}+\frac{1}{2}u_{j}^\top u_{j})^{a_{j}+q_{j}/2}} \E \big[\int_{\mathbb{R}^{p}} \I\big\{\norm{\beta}  \le l \big(\norm{\delta^*} + \norm{Z^*_{\triangle} u}\big)\big\} d\beta \big] du \nonumber\\
&= (2\pi)^{-\frac{q}{2}}\prod_{i \in I_{3}} {m_{i} \choose y_{i}} \int_{\mathbb{R}^{q}} \prod_{j=1}^{r} \frac{\Gamma{(a_{j}+q_{j}/2)}}{(b_{j}+u_{j}^\top u_{j}/2)^{a_{j}+q_{j}/2}} \E  \big[2^p l^p \big(\norm{\delta^*} + \norm{Z^*_{\triangle} u}\big)^p  \big] du \nonumber\\
&\le (2\pi)^{-\frac{q}{2}} \prod_{i \in I_{3}} {m_{i} \choose y_{i}} \int_{\mathbb{R}^{q}} \prod_{j=1}^{r} \frac{\Gamma{(a_{j}+q_{j}/2)}}{(b_{j}+u_{j}^\top u_{j}/2)^{a_{j}+q_{j}/2}} \E  \big[ 2^{2p-1} l^p \big ( \norm{\delta^*}^p + \norm{Z^*_{\triangle} u}^p \big) \big] du \nonumber\\
&= (2\pi)^{-\frac{q}{2}} \prod_{i \in I_{3}} {m_{i} \choose y_{i}} 2^{2p-1} l^p \bigg[\int_{\mathbb{R}^{q}} \prod_{j=1}^{r} \frac{\Gamma{(a_{j}+q_{j}/2)}}{(b_{j}+u_{j}^\top u_{j}/2)^{a_{j}+q_{j}/2}} \E  \norm{\delta^*}^p du \nonumber\\&\hspace{1in} + \int_{\mathbb{R}^{q}} \prod_{j=1}^{r} \frac{\Gamma{(a_{j}+q_{j}/2)}}{(b_{j}+u_{j}^\top u_{j}/2)^{a_{j}+q_{j}/2}}   \norm{Z^*_{\triangle} u}^p  du\bigg] \nonumber \\
&\le (2\pi)^{-\frac{q}{2}} \prod_{i \in I_{3}} {m_{i} \choose y_{i}}   2^{2p-1} l^p  \bigg[ \E \norm{\delta}^p \int_{\mathbb{R}^{q}} \prod_{j=1}^{r} \frac{\Gamma{(a_{j}+q_{j}/2)}}{(b_{j}+u_{j}^\top u_{j}/2)^{a_{j}+q_{j}/2}} du  \nonumber\\&\hspace{1in} + \lambda^{p/2} \int_{\mathbb{R}^{q}} \prod_{j=1}^{r} \frac{\Gamma{(a_{j}+q_{j}/2)}}{(b_{j}+u_{j}^\top u_{j}/2)^{a_{j}+q_{j}/2}}  (\sum_{j=1}^{r} u_{j}^\top u_{j})^{\frac{p}{2}} du \bigg],
\end{align}
where the second inequality follows from the fact that $(a+b)^c \le 2^{c-1} (a^c +  b^c)$ for $c \ge 1$, $a \ge 0$ and $b \ge 0$. Since $\norm{Z^*_{\triangle} u}^p =(u^\top {Z^*_{\triangle}}^{\top} Z^*_{\triangle} u)^{p/2} \le \lambda^{p/2} \norm{u}^p =  \lambda^{p/2} (\sum_{j=1}^{r} u_{j}^\top u_{j})^{\frac{p}{2}}$, where $\lambda$ is the largest eigenvalue of ${Z^*_{\triangle}}^{\top} Z^*_{\triangle}$, the last inequality follows.

Next, we will work on the integration in the first term on the right-hand side of \eqref{eq:binosuff}. Because
\begin{align*}
&\int_{\mathbb{R}^{q}} \prod_{j=1}^{r} \frac{\Gamma{(a_{j}+q_{j}/2)}}{(b_{j}+u_{j}^\top u_{j}/2)^{a_{j}+q_{j}/2}} d u \le \max_{1 \le j \le r} \Gamma{(a_{j}+q_{j}/2)} \prod_{j=1}^{r}\int_{\mathbb{R}^{q_{j}}} \frac{1}{(b_{j}+u_{j}^\top u_{j}/2)^{a_{j}+q_{j}/2}} d u_{j},
\end{align*}
we focus on the following integration
\begin{align}
\label{eq:first_integ}
&\int_{\mathbb{R}^{q_{j}}} \frac{1}{(b_{j}+u_{j}^\top u_{j}/2)^{a_{j}+q_{j}/2}} d u_{j}. 
\end{align}
For $q_{j}=1$, \eqref{eq:first_integ} becomes
$2^{a_{j}+1/2}\int_{0}^{\infty} (2b_{j}+u_{j}^2)^{-a_{j}-1/2} d u_{j}$
which is finite by condition $2$.  For any integer $q_{j} \ge 2$, we
consider the polar transformation as
$u_{j1} = r \cos \theta_{1}, \,u_{j2} = r \sin \theta_{1} \cos
\theta_{2},\dots,u_{jq_{j}-1} = r \sin \theta_{1}\dots \sin
\theta_{q_{j}-2} \cos \theta_{q_{j}-1},\, u_{jq_{j}} = r \sin
\theta_{1}\dots \sin \theta_{q_{j}-2} \sin \theta_{q_{j}-1}$. Here,
$r>0, \, 0<\theta_{q_{j}-1} < 2\pi, \, 0<\theta_{i} < \pi, \, 1 \le i
\le q_{j}-2$ and the Jacobian is
$r^{q_{j}-1} \prod_{i=1}^{q_{j}-2} (\sin
\theta_{i})^{q_{j}-1-i}$. Note that, when $q_{j}=2$, the Jacobian is
$r$. Therefore, \eqref{eq:first_integ} becomes
\begin{align}
\label{eq:first_integ_bound}
&\int_{\mathbb{R}^{q_{j}}} \frac{1}{(b_{j}+u_{j}^\top u_{j}/2)^{a_{j}+q_{j}/2}} d u_{j} \nonumber \\
&=\int_{0}^{2\pi} \int_{0}^{\pi}  \cdots \int_{0}^{\pi} \int_{0}^{\infty} \frac{r^{q_{j}-1} [\I(q_{j} \ge 3)\prod_{i=1}^{q_{j}-2} (\sin \theta_{i})^{q_{j}-1-i}+\I(q_{j}=2)]}{(b_{j}+ r^2/2)^{a_{j}+q_{j}/2}}  \,d r \,d \theta_{1}\,d \theta_{2} \cdots d \theta_{q_{j}-1} \nonumber \\
&\le \int_{0}^{2\pi} \int_{0}^{\pi}  \cdots \int_{0}^{\pi} \int_{0}^{\infty} \frac{r^{q_{j}-1}}{(b_{j}+ r^2/2)^{a_{j}+q_{j}/2}}  \,d r \,d \theta_{1}\,d \theta_{2} \cdots d \theta_{q_{j}-1},
\end{align}
where the inequality is due to the fact $0<\prod_{i=1}^{q_{j}-2} (\sin \theta_{i})^{q_{j}-1-i} \le 1$ for $0<\theta_{i} < \pi, \, 1 \le i \le q_{j}-2$.

Now, we work on the integration for $r$ in
\eqref{eq:first_integ_bound}, considering $a_{j}>0, b_{j} >0$ and
using $r=\sqrt{2b_{j}} \tan \alpha $,
\begin{align*}
\int_{0}^{\infty} \frac{r^{q_{j}-1}}{(b_{j}+ r^2/2)^{a_{j}+q_{j}/2}}  \,d r &= \int_{0}^{\frac{\pi}{2}}
2^{\frac{q_{j}}{2}} b_{j}^{-a_{j}}(\tan \alpha)^{q_{j}-1} (\sec \alpha)^{2-q_{j}-2a_{j}}\, d \alpha \\
&=2^{\frac{q_{j}}{2}} b_{j}^{-a_{j}} \int_{0}^{\frac{\pi}{2}} \tan \alpha \big(\frac{\sin \alpha}{\cos \alpha} \big)^{q_{j}-2} (\sec \alpha)^{2-q_{j}-2a_{j}} \, d \alpha \\
&\le 2^{\frac{q_{j}}{2}} b_{j}^{-a_{j}} \int_{0}^{\frac{\pi}{2}} \tan \alpha  (\sec \alpha)^{-2a_{j}} \, d \alpha = \frac{2^{\frac{q_{j}}{2}-1} b_{j}^{-a_{j}}}{a_{j}}
\end{align*}
where the inequality is due to the fact that $(\sin \alpha)^{q_{j}-2} \le 1$ for $q_{j} \ge 2$ and $\alpha \in [0,\pi/2]$. 

Hence, by combining the results for $q_{j}=1$ and $q_{j} \ge 2$, we
obtain: if $a_{j}>1/2,\, b_{j}>0,$ for $j=1, \dots,r$, and
$\E \norm{\delta}^{p} < \infty$, the first term from the right
hand side of \eqref{eq:binosuff} is finite.

Note that for $p \ge 1$, the integration in the second term on the right-hand side of \eqref{eq:binosuff} is 
\begin{align}
  \label{eq:ujtuj}
&\int_{\mathbb{R}^{q}} \bigg(\sum_{j=1}^{r} u_{j}^\top u_{j}\bigg)^{\frac{p}{2}} \prod_{j=1}^{r} \frac{\Gamma{(a_{j}+q_{j}/2)}}{(b_{j}+u_{j}^\top u_{j}/2)^{a_{j}+q_{j}/2}} d u \nonumber\\
&\le \int_{\mathbb{R}^{q}} \max\big(2^{(r-1)(\frac{p}{2}-1)},1 \big) \bigg[\sum_{j=1}^{r} (u_{j}^\top u_{j})^{\frac{p}{2}} \bigg] \prod_{j=1}^{r} \frac{\Gamma{(a_{j}+q_{j}/2)}}{(b_{j}+u_{j}^\top u_{j}/2)^{a_{j}+q_{j}/2}} d u,
\end{align}
where the inequality is due to the fact that for $a \ge 0$ and $b \ge 0$, $(a+b)^c \le 2^{c-1}( a^c + b^c)$ for $c \ge 1$ and $(a+b)^c \le  a^c + b^c$ for $0<c <1$. Then we focus on one term in the expression on the right-hand side of \eqref{eq:ujtuj},
\begin{align}
\label{eq:oneterm}
&\int_{\mathbb{R}^{q}} \max\big(2^{(r-1)(\frac{p}{2}-1)},1 \big) (u_{j}^\top u_{j})^{\frac{p}{2}} \prod_{j=1}^{r} \frac{\Gamma{(a_{j}+q_{j}/2)}}{(b_{j}+u_{j}^\top u_{j}/2)^{a_{j}+q_{j}/2}} d u \nonumber\\&= \max\big(2^{(r-1)(\frac{p}{2}-1)},1 \big) \int_{\mathbb{R}^{q-q_{j}}} \int_{\mathbb{R}^{q_{j}}}  (u_{j}^\top u_{j})^{\frac{p}{2}} \frac{\Gamma{(a_{j}+q_{j}/2)}}{(b_{j}+u_{j}^\top u_{j}/2)^{a_{j}+q_{j}/2}} d u_{j} \nonumber\\ & \quad \prod_{j' \neq j} \frac{\Gamma{(q_{j'}/2+a_{j'})}}{(b_{j'}+u_{j'}^\top u_{j'}/2)^{q_{j'}/2+a_{j'}}} d u\setminus u_{j}.
\end{align}
Next, for the inner integration with respect to $u_{j}$,
$\int_{\mathbb{R}^{q_{j}}} (u_{j}^\top
  u_{j})^{\frac{p}{2}}/[(b_{j}+u_{j}^\top u_{j}/2)^{a_{j}+q_{j}/2}] d
u_{j},$ again we consider $q_{j}=1$ and $q_{j} \ge 2$ separately. When $q_{j}=1$,
by changing the variable $\sqrt{2b_{j}}\tan \theta = u_{j}$, we obtain
\begin{align*}
\int_{\mathbb{R}^{q_{j}}}  \frac{(u_{j}^\top u_{j})^{\frac{p}{2}}}{(b_{j}+u_{j}^\top u_{j}/2)^{a_{j}+q_{j}/2}}  d u_{j}&=2^{a_{j}+1/2} \int_{0}^{\infty} \frac{u_{j}^p}{(2b_{j}+u_{j}^2)^{a_{j}+1/2}} du_{j}\\ &= 2^{(p+1)/2} b_{j}^{p/2-a_{j}} \int_{0}^{\frac{\pi}{2}} (\sin \theta)^{p} (\sec \theta)^{p+1-2a_{j}} d \theta\\
&= 2^{(p+1)/2} b_{j}^{p/2-a_{j}} \int_{0}^{\frac{\pi}{2}} \tan \theta (\sin \theta)^{p-1} (\sec \theta)^{p-2a_{j}} d \theta\\
&\le 2^{(p+1)/2} b_{j}^{p/2-a_{j}} \int_{0}^{\frac{\pi}{2}} \tan \theta (\sec \theta)^{p-2a_{j}} d \theta\\ &= 2^{(p+1)/2} b_{j}^{p/2-a_{j}}/(2a_{j} -p),
\end{align*}
where the inequality is due to the fact that $(\sin \theta)^{p-1} \le 1$ for $p \ge 1$ and $\theta \in [0,\pi/2]$, and the integration in the last inequality is finite if $a_{j} >p/2$. 

For $q_{j} \ge 2$, we apply the polar transformation again and have
\begin{align*}
&\int_{\mathbb{R}^{q_{j}}}  \frac{(u_{j}^\top u_{j})^{\frac{p}{2}}}{(b_{j}+u_{j}^\top u_{j}/2)^{a_{j}+q_{j}/2}}  d u_{j}\\
 &= \int_{0}^{2\pi} \int_{0}^{\pi}\cdots \int_{0}^{\pi} \int_{0}^{\infty} \frac{r^{p+q_{j}-1} [\I(q_{j} \ge 3)\prod_{i=1}^{q_{j}-2} (\sin \theta_{i})^{q_{j}-1-i}+ \I(q_{j}=2)]}{(b_{j}+ r^2/2)^{a_{j}+q_{j}/2}}  \,d r \,d \theta_{1}\,d \theta_{2} \cdots d \theta_{q_{j}-1}\\
&\le \int_{0}^{2\pi} \int_{0}^{\pi}\cdots \int_{0}^{\pi} \int_{0}^{\infty} \frac{r^{p+q_{j}-1}}{(b_{j}+ r^2/2)^{a_{j}+q_{j}/2}}  \,d r \,d \theta_{1}\,d \theta_{2} \cdots d \theta_{q_{j}-1}.
\end{align*}
As before, using $r=\sqrt{2b_{j}} \tan \alpha$, we have
\begin{align*}
\int_{0}^{\infty} \frac{r^{p+q_{j}-1}}{(b_{j}+ r^2/2)^{a_{j}+q_{j}/2}}  \,d r &= \int_{0}^{\frac{\pi}{2}}
2^{(p+q_{j})/2} b_{j}^{p/2-a_{j}}(\tan \alpha)^{p+q_{j}-1} (\sec \alpha)^{2-q_{j}-2a_{j}}\, d \alpha \\
&=2^{(p+q_{j})/2} b_{j}^{p/2-a_{j}} \int_{0}^{\frac{\pi}{2}} \tan \alpha \frac{(\sin \alpha)^{p+q_{j}-2}}{(\cos \alpha)^{p+q_{j}-2}} (\sec \alpha)^{2-q_{j}-2a_{j}} \, d \alpha \\
                                                                              &\le 2^{(p+q_{j})/2} b_{j}^{p/2-a_{j}} \int_{0}^{\frac{\pi}{2}} \tan \alpha (\sec \alpha)^{p-2a_{j}} \, d \alpha\\
  &=2^{(p+q_{j})/2} b_{j}^{p/2-a_{j}} /(2a_{j}-p).
\end{align*}
The integration for
$\alpha$ in the last inequality is finite if $a_{j} > p/2$. Thus, when
$a_{j} > p/2,\, b_{j}>0$ for $j=1,\dots,r$, the inner integration
respect to $u_{j}$ in \eqref{eq:oneterm} is finite. For the outer
integration respect to $u \setminus u_{j}$ in \eqref{eq:oneterm}, we
can use the conditions for the integration in the first term from
\eqref{eq:binosuff}. Hence if $a_{j}>1/2,\, b_{j}>0$ for
$j=1, \dots,r$, the outer integration respect to $u_{j'}$ in
\eqref{eq:oneterm} is finite. Thus, Theorem \ref{theosuf} is proved.
\end{proof}
\begin{proof}[Proof of Lemma~\ref{suf_poisson}]
   If $\pi(\beta) \propto 1 $, and the prior for $\tau_{j}$ is as in \eqref{eq:tauprior}, $c(y)$ in \eqref{eq:cy_poisson} becomes
\begin{align}
\label{eq:cy_poisson2}
c(y)& = (2\pi)^{-\frac{q}{2}}\int_{\mathbb{R}_{+}^{r}} \int_{\mathbb{R}^{p}} \int_{\mathbb{R}^{q}} \prod_{i=1}^{n} \frac{\exp [(x_{i}^\top \beta + z_{i}^\top u)y_{i}]}{\exp[\exp(x_{i}^\top \beta + z_{i}^\top u)] y_{i} !} \prod_{j=1}^{r} \tau_{j}^{a_{j}+q_{j}/2-1}\nonumber \\
&\quad \exp[-\tau_{j}(b_{j} + u_{j}^\top u_{j}/2)] d u d\beta d\tau \nonumber \\
& = (2\pi)^{-\frac{q}{2}}B\int_{\mathbb{R}_{+}^{r}} \int_{\mathbb{R}^{p}} \int_{\mathbb{R}^{q}} \prod_{i=1}^{n} {y_{(n)} \choose y_{i}}    \frac{\exp [(x_{i}^\top \beta + z_{i}^\top u)y_{i}]}{\exp[\exp(x_{i}^\top \beta + z_{i}^\top u)]} \prod_{j=1}^{r} \tau_{j}^{a_{j}+q_{j}/2-1} \nonumber \\
&\quad \exp[-\tau_{j}(b_{j} + u_{j}^\top u_{j}/2)] d u d\beta d\tau \nonumber \\
& \le (2\pi)^{-\frac{q}{2}}B\int_{\mathbb{R}_{+}^{r}} \int_{\mathbb{R}^{p}} \int_{\mathbb{R}^{q}} \prod_{i=1}^{n} {y_{(n)} \choose y_{i}}    \frac{d \exp [(x_{i}^\top \beta + z_{i}^\top u)y_{i}]}{[1+\exp(x_{i}^\top \beta + z_{i}^\top u)]^{y_{(n)}}} \prod_{j=1}^{r} \tau_{j}^{a_{j}+q_{j}/2-1} \nonumber \\
&\quad \exp[-\tau_{j}(b_{j} + u_{j}^\top u_{j}/2)] d u d\beta d\tau,
\end{align}
where $B = \prod_{i=1}^{n}(y_{(n)}-y_{i})!/y_{(n)}!$ and the last inequality is because there exists a constant $d$ such that 
\begin{align}
\label{eq:inequality_exp_ch2}
(1+\exp(w))^{y_{(n)}} \le d\exp[\exp(w)]
\end{align}
 for $w \in \mathbb{R}$. We now prove \eqref{eq:inequality_exp_ch2}. When $y_{(n)}=0$, it is straightforward to see that  \eqref{eq:inequality_exp_ch2} is satisfied with $d=1$. When $y_{(n)}=1$, note that $\exp[\exp(w)] \ge 1+\exp(w)$ since $\exp[\exp(w)] -\exp(w)-1$ is an increasing function and $\lim_{w \to -\infty}\exp[\exp(w)] -\exp(w)-1 =0$. If $y_{(n)} \ge 2$, let $ g(w)=\exp(w) - y_{(n)} \log(1+\exp(w))$, then $g'(w) = [\exp(2w)-\exp(w)(y_{(n)}-1)]/[1+\exp(w)]$. Note that $g'(w) \lesseqqgtr 0$ if and only if $\exp(w) \lesseqqgtr y_{(n)}-1 $, that is, $w \lesseqqgtr \log(y_{(n)}-1)$. Hence, $g(w) \ge g(\log(y_{(n)}-1))$, that is $\exp(w) \ge y_{(n)} \log(1+\exp(w)) + g(\log(y_{(n)}-1))$. Thus, we have \eqref{eq:inequality_exp_ch2}, where $d =\exp(-g(\log(y_{(n)}-1)))=\exp(1-y_{(n)})y_{(n)}^{y_{(n)}} $. 

Also, from \eqref{eq:cy_poisson2}, we obtain 
\begin{align}
\label{eq:cy_poisson3}
c(y) &\le B d(2\pi)^{-\frac{q}{2}}\int_{\mathbb{R}_{+}^{r}} \int_{\mathbb{R}^{p}} \int_{\mathbb{R}^{q}} \prod_{i=1}^{n} {y_{(n)} \choose y_{i}}    \big[F_L(x_{i}^\top \beta + z_{i}^\top u)\big]^{y_{i}} \big[1-F_L(x_{i}^\top \beta + z_{i}^\top u)\big]^{y_{(n)}-y_{i}} \nonumber\\
& \quad \prod_{j=1}^{r} \tau_{j}^{a_{j}+q_{j}/2-1} \exp[-\tau_{j}(b_{j} + u_{j}^\top u_{j}/2)] d u d\beta d\tau,
\end{align}
where $F_L(t) = e^{t}/(1+e^{t})$ is the cdf of the standard logistic random variable. Now, we can observe the integrand in \eqref{eq:cy_poisson3} is the same as that in \eqref{eq:cy} when the prior on $\beta$ is $\pi(\beta) \propto 1$, the prior on $\tau$ is as in \eqref{eq:tauprior} and $F \equiv F_L$. Since the standard logistic random variable has all finite moments, the Lemma is proved based on Theorem \ref{theosuf}.
\end{proof}
\begin{proof}[Proof of Theorem~\ref{theosuf_powerprior}]
  Using \eqref{eq:cy_binomial_ineq} and $b_{j}=0$, we have  
\begin{align}
\label{eq:cy_binomialz_powerprior}
c(y) &\le (2\pi)^{-\frac{q}{2}}\int_{\mathbb{R}_{+}^{r}} \int_{\mathbb{R}^{q}}  \bigg[\prod_{i \in I_{3}} {m_{i} \choose y_{i}}\bigg] \int_{\mathbb{R}^{p}} \E\,\big[\I\{t_{i} (x_{i}^\top\beta + z_{i}^\top u) \le t_{i} \delta_{i};1 \le i \le n+k \} \big]d\beta  \nonumber\\ &\quad\prod_{j=1}^{r} \tau_{j}^{a_{j}+q_{j}/2-1} \exp(-\tau_{j} u_{j}^\top u_{j}/2) d u d \tau \nonumber\\
&= (2\pi)^{-\frac{q}{2}}\int_{\mathbb{R}_{+}^{r}} \int_{\mathbb{R}^{q}}  \big[\prod_{i \in I_{3}} {m_{i} \choose y_{i}}\big] \E \big[\int_{\mathbb{R}^{p}} \I\{(X^*_\triangle,Z^*_\triangle) (\beta^\top,u^\top)^\top \le \delta^* \} d\beta \big] \nonumber \\
& \quad \prod_{j=1}^{r} \tau_{j}^{a_{j}+q_{j}/2-1} \exp(-\tau_{j} u_{j}^\top u_{j}/2)d u d \tau \nonumber \\
&\le (2\pi)^{-\frac{q}{2}}\int_{\mathbb{R}_{+}^{r}} \int_{\mathbb{R}^{q}}  \big[\prod_{i \in I_{3}} {m_{i} \choose y_{i}}\big] \E \big[\int_{\mathbb{R}^{p}} \I \big\{\norm{ (\beta^\top,u^\top)} \le l' \norm{\delta^*} \big\} d\beta \big] \nonumber \\
& \quad \prod_{j=1}^{r} \tau_{j}^{a_{j}+q_{j}/2-1} \exp(-\tau_{j} u_{j}^\top u_{j}/2)d u d \tau \nonumber \\
&\le (2\pi)^{-\frac{q}{2}}\int_{\mathbb{R}_{+}^{r}} \int_{\mathbb{R}^{q}}  \big[\prod_{i \in I_{3}} {m_{i} \choose y_{i}}\big] \E \big[2^p l'^p \norm{\delta^*}^p \I\{\norm{ u} \le l' \norm{\delta^*} \}  \big] \nonumber \\
& \quad \prod_{j=1}^{r} \tau_{j}^{a_{j}+q_{j}/2-1} \exp(-\tau_{j} u_{j}^\top u_{j}/2)d u d \tau \nonumber \\
&\le \kappa \E \bigg[ \norm{\delta^*}^p \int_{\mathbb{R}_{+}^{r}} \prod_{j=1}^{r}\int_{\mathbb{R}^{q_{j}}} \I\{\norm{ u_{j}} \le l' \norm{\delta^*} \}    \tau_{j}^{a_{j}+q_{j}/2-1} \exp(-\tau_{j} u_{j}^\top u_{j}/2)d u_j d \tau \bigg],
\end{align}
where $\kappa$ is a constant. Here, we have used the condition $1$ and
\pcite{chen2001propriety} Lemma $4.1$ to obtain the second inequality,
where $l'$ depends on $(X^*_\triangle,Z^*_\triangle)$.  Note that
$\delta=(\delta_{1},\delta_{2},...,\delta_{n+k})$,
$\delta^* = (t_{1}\delta_{1}, t_{2}\delta_{2},\dots,t_{n+k}
\delta_{n+k})^\top$, where $t_{i}$ is defined as before, and
$\norm{\delta^*}=\norm{\delta}$. Then from
\eqref{eq:cy_binomialz_powerprior}, applying similar techniques as in
the proof of Theorem $4.2$ in \cite{chen2002necessary}, we have
\begin{align*}
c(y) &\le \kappa \,\E \big[ \norm{\delta}^p \int_{\mathbb{R}_{+}^{r}} \prod_{j=1}^{r}\tau_{j}^{a_{j}-1} \int_{\mathbb{R}^{q_{j}}} \I \{\norm{ u_{j}} \le l' \norm{\delta} \}    \tau_{j}^{q_{j}/2} \exp(-\tau_{j} u_{j}^\top u_{j}/2)d u_j d \tau \big] \nonumber\\
&\le \kappa_1 \,\E \big[ \norm{\delta}^p \int_{\mathbb{R}_{+}^{r}} \prod_{j=1}^{r}\tau_{j}^{a_{j}-1} \min \big(1,2^{q_{j}/2} \pi^{-q_{j}/2}l'^{q_{j}}\tau_{j}^{q_{j}/2} \norm{\delta}^{q_{j}} \big) d \tau \big] \nonumber\\
     &= \kappa_1 \,\E \big\{ \norm{\delta}^p \prod_{j=1}^{r} \big[\int_{l_{1}}^{\infty} \tau_{j}^{a_{j}-1} d \tau_{j} + 2^{q_{j}/2} \pi^{-q_{j}/2}l'^{q_{j}} \norm{\delta}^{q_{j}} \int_{0}^{l_{1}} \tau_{j}^{a_{j}+q_{j}/2-1}   d \tau_{j} \big]\big\}\\
  & <\infty,
\end{align*}
where $\kappa_1$ is a constant, $l_{1} = \pi/(2l'^2 \norm{\delta}^2)$ and the integrations in the last line are finite as the conditions $2$ and $3$ hold. Therefore, the Theorem is proved.

\end{proof}
\begin{proof}[Proof of Theorem~\ref{theo_nece1}]
When we have binomial responses, from \eqref{eq:cy} we have
\begin{align}
\label{eq:cy_nec}
c(y)&=(2\pi)^{-\frac{q}{2}} \int_{\mathbb{R}_{+}^{r}} \int_{\mathbb{R}^{p}} \int_{\mathbb{R}^{q}}  \bigg[\prod_{i=1}^{n} {m_{i} \choose y_{i}} \big{[}F(x_{i}^\top \beta + z_{i}^\top u)\big{]}^{y_{i}}\big{[}1- F(x_{i}^\top \beta + z_{i}^\top u) \big{]}^{m_{i}-y_{i}}\bigg] \nonumber \\
& \quad \bigg[ \prod_{j=1}^{r} \tau_{j}^{a_{j}+q_{j}/2 -1} \exp{[-\tau_{j} (b_{j}+u_{j}^\top u_{j}/2)]}\bigg] du d\beta d\tau.
\end{align}
If $X$ is not a full column rank matrix, there exists $\beta^* \neq 0$
such that $X \beta^* =0$, that is $x_{i}^\top \beta^* =0$ for
$i=1,\dots,n$.  For $\epsilon' >0$, define
$D_{\epsilon'} =\big\{\tilde{\beta} \in \mathbb{R}^p: \abs{x_{i}^\top
  \tilde{\beta}} < \epsilon', \, i=1,\dots,n \big\} $.  Recall that
$F$ is a nondecreasing function. Then we have
\begin{align}
\label{eq:cy_nec2}
c(y) &\ge (2\pi)^{-\frac{q}{2}} \int_{\mathbb{R}_{+}^{r}} \int_{\mathbb{R}^{q}}  \prod_{i=1}^{n} {m_{i} \choose y_{i}} \big[F(-\epsilon' + z_{i}^\top u)\big]^{y_{i}}\big{[}1- F(\epsilon' + z_{i}^\top u) \big{]}^{m_{i}-y_{i}} \nonumber \\
&\quad \prod_{j=1}^{r} \tau_{j}^{a_{j}+q_{j}/2 -1} \exp{[-\tau_{j} (b_{j}+u_{j}^\top u_{j}/2)]} du  d\tau \int_{\beta \in D_{\epsilon'}} d\beta.
\end{align}
Note that
$ \sum_{i=1}^{n} (x_{i}^\top \tilde{\beta})^2 < \epsilon (=
\epsilon'^2) \Rightarrow \tilde{\beta} \in D_{\epsilon'}.$ Now
$ \sum_{i=1}^{n} (x_{i}^\top \tilde{\beta})^2 < \epsilon
\Leftrightarrow \tilde{\beta}^\top X^\top X \tilde{\beta} <
\epsilon$. By spectral decomposition for $X^\top X$, we have
$X^\top X = P \Lambda P^\top$ where $\Lambda$ is the $p \times p$
diagonal matrix with eigenvalues $\lambda_{1},\dots,\lambda_{p}$ and
the $i$th eigenvector of $X^\top X$ is $p_{i}^\top$, the $i$th row of
$P^\top$. Thus,
\begin{align}
\label{eq:general_xbeta_prime}
\tilde{\beta}^\top P \Lambda P^\top \tilde{\beta} < \epsilon \Leftrightarrow \lambda_{1} g_{1}^2 + \cdots + \lambda_{p} g_{p}^2 < \epsilon,
\end{align}
where $g_{i} = p_{i}^\top \tilde{\beta}$ for $i=1,\dots,p$. Since $X$
is not a full column rank matrix, $\rank(X) \le p-1$. Suppose $\rank(X)=p-1$,
then without loss of generality, let $\lambda_{p}=0$ and
$\lambda_{1},\dots,\lambda_{p-1}$ are all positive. Hence,
\eqref{eq:general_xbeta_prime} becomes
$\lambda_{1} g_{1}^2 + \cdots + \lambda_{p-1} g_{p-1}^2 < \epsilon$.
Let $g=(g_{1},\dots,g_{p})^\top$. Let
$B_1 = \big\{g \in \mathbb{R}^p: \abs{g_{i}} < \sqrt{\epsilon/\tr(X^\top X)}, 1 \le i\le p-1, \, g_{p} \in \mathbb{R}
\big\}$. By change of variables $\beta \rightarrow g$, we have $\int_{\beta \in D_{\epsilon'}} d\beta = \int_{g \in B_1} d g = \infty$. Note that the Jacobian of the transformation $\beta \rightarrow g$ is $\det(P)$ and $\abs{\det (P)}=1$. Since the integrand in the multiple of $\int_{\beta \in D_{\epsilon'}} d\beta$ on the right-hand side of \eqref{eq:cy_nec2} is nonnegative and is not zero (a.e.), the integral is strictly positive. Therefore, $c(y)$ diverges to infinity.

As for the second necessary condition, we focus on the following part in \eqref{eq:cy_nec}:
\begin{align*}
\int_{\mathbb{R}_{+}^{r}}  \prod_{j=1}^{r} \tau_{j}^{a_{j}+q_{j}/2 -1} \exp{[-\tau_{j} (b_{j}+u_{j}^\top u_{j}/2)]}  d\tau =\prod_{j=1}^{r} \int_{\mathbb{R}_{+}} \tau_{j}^{a_{j}+q_{j}/2 -1} \exp{[-\tau_{j} (b_{j}+u_{j}^\top u_{j}/2)]}  d\tau_{j} .
\end{align*}
For fixed $u_{j}$, if $b_{j}+u_{j}^\top u_{j}/2 \le 0$, then we obtain $\int_{0}^{\infty} \tau_{j}^{a_{j}+q_{j}/2 -1} \exp{[-\tau_{j} (b_{j}+u_{j}^\top u_{j}/2)]}  d\tau_{j} \ge  \int_{0}^{\infty} \tau_{j}^{a_{j}+q_{j}/2 -1} d\tau_{j}= \infty $. Also, when $b_{j}+u_{j}^\top u_{j}/2>0$ and $a_{j}+q_{j}/2 \le 0$, we have
\[
  \int_{0}^{\infty} \tau_{j}^{a_{j}+q_{j}/2 -1} \exp{[-\tau_{j} (b_{j}+u_{j}^\top u_{j}/2)]}  d\tau_{j} \ge    \exp{[- (b_{j}+u_{j}^\top u_{j}/2)]} \int_{0}^{1} \tau_{j}^{a_{j}+q_{j}/2 -1}   d\tau_{j} = \infty .
\]
Consequently, the second necessary condition for posterior propriety is proved. 
\end{proof}

\begin{proof}[Proof of Theorem~\ref{theo_nece2}]
Recall $X^*_\triangle$ as defined in Section \ref{subs:binomial}. As in the proof of Theorem 2 of \cite{roy2013posterior} \cite[also see the proof of Theorem 2.2 in][]{chen2001propriety}, $D=\{X^{*\top}_\triangle a: a >0, a \in \mathbb{R}^{n+k}\}$ is a convex cone in $\mathbb{R}^p$. If the condition 2 is not satisfied, by Corollary 11.7.3 of \cite{tyrrell1970convex}, there exists some non-zero vector $h=(h_{1},\dots,h_{p})^\top  \in \mathbb{R}^{p}$ such that $h^\top X^{*\top}_\triangle a \le 0$, for all  $a > 0$. Hence, $\forall \; a \ge 0,\; h^\top X^{*\top}_\triangle a \le 0$ and we obtain
\begin{align}
\label{eq:goodre}
t_{i} x_{\triangle i}^\top h \le 0, \,i=1,2,\dots,n+k,
\end{align}
where $t_{i}$ and $x_{\triangle i}^\top$ are defined in Section \ref{subs:binomial}. Then we have 
\begin{align}
\label{eq:xh}
x_{\triangle i}^\top h \le 0 \;\text{for}\; i \in I_{1}; \quad
x_{\triangle i}^\top h \ge 0 \;\text{for}\; i \in I_{2}; \quad
x_{\triangle i}^\top h = 0 \;\text{for}\; i \in I_{3}. 
\end{align}
Note that $x_{\triangle i}^\top h = 0 \;\text{for}\; i \in I_{3}$ because $t_i = 1$ for $i \in I_{3}$ and also $t_i=-1$ for $i= n+1,\dots n+k$. Without loss of generality, assume that $h_{1} \ne 0$, $\beta = s_{1} h + (0,s_{2},\dots,s_{p})^\top $ and $s =(s_{1},\dots,s_{p})^\top$. For fixed $a>0$, let us define $B_{2} =\big\{u \in \mathbb{R}^q: -a \le z_{i}^\top u \le a,\,i = 1,\dots, n \big\} $. Since $Z$ has full column rank, from \eqref{eq:cy_nec}, we have
\begin{align}
\label{eq:cy_nec3}
c(y) &=(2\pi)^{-\frac{q}{2}} \int_{\mathbb{R}_{+}^{r}} \int_{\mathbb{R}^{p}} \int_{\mathbb{R}^{q}}  \bigg[\prod_{i=1}^{n} {m_{i} \choose y_{i}} \big{[}F(x_{i}^\top \beta + z_{i}^\top u)\big{]}^{y_{i}}\big{[}1- F(x_{i}^\top \beta + z_{i}^\top u) \big{]}^{m_{i}-y_{i}} \bigg]\nonumber \\ &\quad \bigg[\prod_{j=1}^{r} \tau_{j}^{a_{j}+q_{j}/2 -1} \exp{[-\tau_{j} (b_{j}+u_{j}^\top u_{j}/2)]}\bigg] du d\beta d\tau \nonumber\\
&\ge (2\pi)^{-\frac{q}{2}} \int_{\mathbb{R}_{+}^{r}} \int_{\mathbb{R}^{p}} \int_{u \in B_{2}}  \bigg[ \prod_{i=1}^{n} {m_{i} \choose y_{i}} \big{[}F(x_{i}^\top \beta -a)\big{]}^{y_{i}}\big{[}1- F(x_{i}^\top \beta + a) \big{]}^{m_{i}-y_{i}} \bigg]\nonumber\\ & \quad \bigg[\prod_{j=1}^{r} \tau_{j}^{a_{j}+q_{j}/2 -1} \exp{[-\tau_{j} (b_{j}+u_{j}^\top u_{j}/2)]}\bigg] du d\beta d\tau \nonumber\\
&=\abs{h_{1}} (2\pi)^{-\frac{q}{2}} \int_{\mathbb{R}_{+}^{r}} \int_{\mathbb{R}^{p}} \int_{u \in B_{2}} \prod_{i \in I_{1}} \big[1- F(s_{1} x_{i}^\top h+ x_{i}^\top (0,s_{2},\dots,s_{p})^\top + a) \big{]}^{m_{i}} \nonumber \\& \quad  \prod_{i \in I_{2}} \big[ F(s_{1} x_{i}^\top h+ x_{i}^\top (0,s_{2},\dots,s_{p})^\top-a)\big]^{m_{i}} \nonumber \\& \quad \prod_{i \in I_{3}} {m_{i} \choose y_{i}} \big[F(x_{i}^\top (0,s_{2},\dots,s_{p})^\top-a)\big{]}^{y_{i}}\big[1- F( x_{i}^\top (0,s_{2},\dots,s_{p})^\top + a) \big]^{m_{i}-y_{i}} \nonumber \\& \quad \prod_{j=1}^{r} \tau_{j}^{a_{j}+q_{j}/2 -1} \exp{[-\tau_{j} (b_{j}+u_{j}^\top u_{j}/2)]} du d s d\tau,
\end{align}
where the above inequality is based on the definition of the set $B_{2}$ and the last equality follows from a change of variables $\beta \rightarrow s$ with the Jacobian of the transformation being $h_{1}$.  

For fixed $r_1 >0$, define $B_{3} =\big\{s \in \mathbb{R}^p: s_{1} \ge 0,\, \abs{s_{k}} \le r_{1},\, 2 \le k \le p \big\} $. By Cauchy-Schwarz inequality, for $s \in B_{3}$, we have 
\begin{align}
\label{eq:cs_nec}
\abs{x_{i}^\top (0,s_{2},\dots,s_{p})^\top} \le \norm{x_{i}} \sqrt{(p-1) r_{1}^2} \le \norm{x_{i}} p r_{1}
\end{align}
From \eqref{eq:xh}, when $i \in I_{2} \cup I_{3}$,
$x_{i}^\top h \ge 0$. If $s_{1} \ge 0$, then we have
$s_{1} x_{i}^\top h \ge 0$ and thus, for $s \in B_{3}$, 
$s_{1} x_{i}^\top h+ x_{i}^\top (0,s_{2},\dots,s_{p})^\top \ge
-\norm{x_{i}} p r_{1}$. Then for
$i \in I_{2} \cup I_{3}$, we have
\begin{align}
\label{eq:FJ}
F \big(s_{1} x_{i}^\top h+ x_{i}^\top (0,s_{2},\dots,s_{p})^\top-a \big) \ge  F \big(-\norm{x_{i}} p r_{1}-a \big).
\end{align}
Similarly from \eqref{eq:xh}, when $i \in I_{1} \cup I_{3}$, $x_{i}^\top h \le 0$. If $s_{1} \ge 0$, then we have $s_{1} x_{i}^\top h \le 0$. Since \eqref{eq:cs_nec} holds, we have $s_{1} x_{i}^\top h+ x_{i}^\top (0,s_{2},\dots,s_{p})^\top \le \norm{x_{i}} p r_{1}$. Thus, for $i \in I_{1} \cup I_{3}$, we obtain
\begin{align}
\label{eq:FI}
 1-F \big(s_{1} x_{i}^\top h+ x_{i}^\top (0,s_{2},\dots,s_{p})^\top+a \big)  \ge  1-F\big(\norm{x_{i}} p r_{1}+a \big).
\end{align}

Applying \eqref{eq:FJ}, \eqref{eq:FI}, from \eqref{eq:cy_nec3}, we have
\begin{align*}
c(y) &\ge \abs{h_{1}} (2\pi)^{-\frac{q}{2}} \prod_{i \in I_{3}} {m_{i} \choose y_{i}} \int_{\mathbb{R}_{+}^{r}} \int_{s \in B_{3}} \int_{u \in B_{2}}  \prod_{i \in I_{2} \cup I_{3}} \big[F \big(-\norm{x_{i}} p r_{1}-a \big)\big]^{y_{i}}  \\
& \quad \prod_{i \in I_{1} \cup I_{3}} \big{[}1- F \big(\norm{x_{i}} p r_{1}+a \big) \big{]}^{m_{i}-y_{i}}\prod_{j=1}^{r} \tau_{j}^{a_{j}+q_{j}/2 -1} \exp{[-\tau_{j} (b_{j}+u_{j}^\top u_{j}/2)]} du d s d\tau\\
&\ge \abs{h_{1}} (2\pi)^{-\frac{q}{2}} \prod_{i \in I_{3}} {m_{i} \choose y_{i}} \int_{\mathbb{R}_{+}^{r}} \int_{u \in B_{2}}  \prod_{i \in I_{2} \cup I_{3}} \big{[}F(- r_{2}-a)\big{]}^{y_{i}}  \prod_{i \in I_{1} \cup I_{3}} \big{[}1- F(r_{2}+a) \big{]}^{m_{i}-y_{i}}\\
& \quad \prod_{j=1}^{r} \tau_{j}^{a_{j}+q_{j}/2 -1} \exp{[-\tau_{j} (b_{j}+u_{j}^\top u_{j}/2)]} du  d\tau  \int_{s \in B_{4}}  d s = \infty,
\end{align*}
where fixed $r_{1}$ and $r_{2}$ are such that $B_{4} \equiv \big\{s \in \mathbb{R}^p: s_{1} \ge 0,\, \abs{s_{k}} \le r_{1},\, 2 \le k \le p, \, \max_{1 \le i \le n}\norm{x_{i}} pr_{1} \le r_{2}\big\} $ is nonempty. 
\end{proof}

\begin{proof}[Proof of Theorem~\ref{theonec_expo}]
We must show that $c(y)$ in \eqref{eq:cy_expo} diverges to infinity if $X$ is not full column rank. Now, if $X$
is not a full column rank matrix, then there exists $\beta^* \neq 0$ such
that $X \beta^* =0$, that is $x_{i}^\top \beta^* =0$ for
$i=1,\dots,n$.  For $\epsilon' >0$, define
$D_{\epsilon'} =\big\{\tilde{\beta} \in \mathbb{R}^p: \abs{x_{i}^\top
  \tilde{\beta}} < \epsilon', \, i=1,\dots,n \big\} $. 
Recall that $\theta$ is a monotone function,
$\eta_{i}=x_{i}^\top \beta + z_{i}^\top u$ and $\E_{u}$ indicates the expectation with respect to the marginal distribution of $u$. Since $b(\theta)$
is a monotone function, we have
\begin{align}
\label{eq:general_necessary}
c(y) &= \int_{\mathfrak{A}} \int_{\mathbb{R}^{p}}
\int_{\mathbb{R}^{q}}\bigg[\prod_{i=1}^{n} \exp[ (y_{i} \theta_{i} - b(\theta_{i})) + d(y_{i})]\bigg] \phi_{q}(u;0,\Psi) du \pi(\Psi)
d\beta \,d\Psi \nonumber \\
&\propto \int_{\mathfrak{A}} \int_{\mathbb{R}^{p}}
\int_{\mathbb{R}^{q}}\bigg[\prod_{i=1}^{n} \exp[ (y_{i} \theta_{i} - b(\theta_{i}))]\bigg] \phi_{q}(u;0,\Psi) du \pi(\Psi)
d\beta \,d\Psi \nonumber \\
&=  \E_{u} \bigg[\int_{\mathbb{R}^{p}} \prod_{i=1}^{n} \exp[ (y_{i} \theta(\eta_{i}) - b(\theta(\eta_{i}))) ] d\beta\bigg] \nonumber \\
&\ge \E_{u}  \bigg[\prod_{i=1}^{n} \exp[ (y_{i} \theta(\pm \epsilon' + z_{i}^\top u) - b(\theta(\pm \epsilon' + z_{i}^\top u))) ]\bigg] \int_{\beta \in D_{\epsilon'}} d\beta.
\end{align}
Here, $b(\theta(\pm \epsilon' + z_{i}^\top u))$ is
$b(\theta(\epsilon' + z_{i}^\top u))$ if $b(\cdot)$ is increasing and
is $b(\theta(-\epsilon' + z_{i}^\top u))$ if $b(\cdot)$ is
decreasing. Similarly, $y_{i} \theta(\pm \epsilon' + z_{i}^\top u)$ is
$y_{i} \theta(\epsilon' + z_{i}^\top u)$ if $y_i \le 0$ and is
$y_{i} \theta(- \epsilon' + z_{i}^\top u)$ if $y_i >0$. As in the proof of Theorem \ref{theo_nece1}, $\int_{\beta \in D_{\epsilon'}} d\beta = \infty$ and $c(y)$ diverges to infinity. Therefore, the proof is complete.
\end{proof}

\begin{proof}[Details on comparison of tails of $\pi_{J}(\tau)$ and $\pi_{NK}(\tau)$]
  In Section~\ref{sec:natajeff}, we presented $\pi_{J}(\tau)$ and
  $\pi_{NK}(\tau)$ for binary and Poisson GLMMs when $p= q=q_{1}=r=1$.
  To compare the behavior of tails of $\pi_{J}(\tau)$ and
  $\pi_{NK}(\tau)$, let
  $f(\tau)=\log \pi_{J}(\tau) - \log \pi_{NK}(\tau)$. For the binary
  GLMMs, the derivative of $f(\tau)$
\begin{align*}
f'(\tau) &= \frac{-\tau^{-3}+(\tau + n e^{\hat{\beta}}/(1+ e^{\hat{\beta}})^2)^{-3}}{\tau^{-2} - (\tau + n e^{\hat{\beta}}/(1+ e^{\hat{\beta}})^2)^{-2}} + \frac{n e^{\hat{\beta}}}{(1+e^{\hat{\beta}})^2+ n e^{\hat{\beta}} \tau} \nonumber \\
&=\frac{\tau^3 - (\tau + n e^{\hat{\beta}}/(1+ e^{\hat{\beta}})^2)^3}{\tau (\tau + n e^{\hat{\beta}}/(1+ e^{\hat{\beta}})^2) (2\tau + n e^{\hat{\beta}}/(1+ e^{\hat{\beta}})^2) n e^{\hat{\beta}}/(1+ e^{\hat{\beta}})^2} + \frac{n e^{\hat{\beta}}}{(1+ e^{\hat{\beta}})^2+ n e^{\hat{\beta}} \tau} \nonumber \\
         &=\frac{-\frac{3 \tau^2 n e^{\hat{\beta}}}{(1+e^{\hat{\beta}})^2}-\frac{3 \tau\, n^2 e^{\hat{2\beta}}}{(1+e^{\hat{\beta}})^4}-\frac{ n^3 e^{\hat{3\beta}}}{(1+e^{\hat{\beta}})^6}-\frac{ n^3 e^{\hat{5\beta}}}{(1+e^{\hat{\beta}})^6}-\frac{3 \tau^2 n e^{\hat{3\beta}}}{(1+e^{\hat{\beta}})^2}-\frac{3 \tau\, n^2 e^{\hat{4\beta}}}{(1+e^{\hat{\beta}})^4}-\frac{2 n^3 e^{\hat{4\beta}}}{(1+e^{\hat{\beta}})^6}-\frac{6 \tau^2 n e^{\hat{2\beta}}}{(1+e^{\hat{\beta}})^2}-\frac{6 \tau\, n^2 e^{\hat{3\beta}}}{(1+e^{\hat{\beta}})^4}-\frac{\tau^3 n^2 e^{\hat{2\beta}}}{(1+e^{\hat{\beta}})^2}}{\tau (\tau + n e^{\hat{\beta}}/(1+ e^{\hat{\beta}})^2) (2\tau + n e^{\hat{\beta}}/(1+ e^{\hat{\beta}})^2) n e^{\hat{\beta}}/(1+ e^{\hat{\beta}})^2 ((1+ e^{\hat{\beta}})^2+ n e^{\hat{\beta}} \tau)} \\
  &<0.
\end{align*}
Similarly, for the Poisson GLMMs,
 \begin{align*}
f'(\tau) &= \frac{-\tau^{-3}+(\tau+ n e^{\hat{\beta}})^{-3}}{\tau^{-2}-(\tau+ n e^{\hat{\beta}})^{-2}} + \frac{n e^{\hat{\beta}}}{1+n e^{\hat{\beta}} \tau} \nonumber\\
         &=\frac{-3\tau^2ne^{\hat{\beta}}-3\tau\, n^2 e^{2\hat{\beta}}-n^3 e^{3\hat{\beta}}-n^2 e^{2\hat{\beta}} \tau^3}{(\tau^{-2}-(\tau+ n e^{\hat{\beta}})^{-2})(1+n e^{\hat{\beta}}\tau)(\tau +n e^{\hat{\beta}} )^3\tau^3}<0.
 \end{align*}
 Thus, $f(\tau)$ is decreasing in both cases. Also, simple
 calculations show that $f(\tau) \rightarrow \infty$ as
 $\tau \rightarrow 0$ and $f(\tau) \rightarrow - \infty$ as
 $\tau \rightarrow \infty$. Thus, there is $\tau_{0} >0$ with
 $f(\tau_0) = 0$ such that $\pi_{J}(\tau) \gtreqless \pi_{NK}(\tau) $
 if $\tau \lesseqgtr \tau_{0}$.  In Table~\ref{tab:jeff}, we provide
 the values of $\tau_0$ such that $f(\tau_0) \approx 0$ for $n=30$ and
 various values of $\hat{\beta}$. We observe that for the Poisson
 GLMM, $\tau_0$ can be quite large as $\hat{\beta}$ increases implying
 that practically $\pi_J$ is more diffuse than $\pi_{NK}$.
\begin{center}
\begin{table*}[h]
  \caption{Values of $\tau_0$ for different $\hat{\beta}$ values.}
  \centering
\begin{tabular}{c|cccccccc}
  \hline\hline
  &  $\hat{\beta}$&  -2  & -1.5& -1.0& -0.5& 0& 0.5&1.0\\
 \hline
Binary GLMM & $\tau_0$&26&81&196&341&411&341&
                                            196\\\hline
Poisson GLMM &$\tau_0$&61&291&1326&5996&26956&120931&
542186\\
\hline
\end{tabular}
\label{tab:jeff}
\end{table*}
\end{center}
\end{proof}

\end{appendix}
\smallskip
\noindent{\bf Conflict Of Interest.} On behalf of all authors, the corresponding author states
that there is no conflict of interest.
\smallskip

\noindent{\bf Funding.} The second author's work was partially supported by USDA NIFA Grant 2023-70412-41087.

\noindent{\bf Acknowledgements.} The authors acknowledge helpful comments of an anonymous
Associate Editor and an anonymous referee on an earlier version of this article. 


\bibliographystyle{ims}
\bibliography{Ref}

%
%


\end{document}